\documentclass[aps,twocolumn,superscriptaddress,nofootinbib,a4paper,longbibliography]{revtex4-2}
\usepackage{times}
\usepackage{graphicx}
\usepackage{epsfig}
\usepackage{subfigure}
\usepackage{amsfonts}
\usepackage{amsmath}
\usepackage{amssymb}
\usepackage{dsfont}
\usepackage{amsthm}
\usepackage{color}
\usepackage[colorlinks=true,linkcolor=blue,citecolor=blue,urlcolor=blue]{hyperref}
\usepackage{comment}
\usepackage{braket}
\usepackage{physics}
\usepackage{multirow}
\usepackage{float}

\usepackage[latin9]{inputenc}
\usepackage{amsmath}
\usepackage{amssymb, enumerate}
\usepackage{color}
\usepackage{babel}
\usepackage{graphicx}
\usepackage{epstopdf}
\usepackage{float}
\usepackage{thmtools}
\usepackage{thm-restate}
\usepackage{hyperref}
\usepackage{cleveref}
\usepackage{enumerate}
\usepackage{lipsum}
\usepackage{fixmath}
\usepackage{dsfont}

\usepackage{amsthm}
\usepackage{amsmath}
\usepackage{amssymb}
\usepackage{tikz}
\usepackage{color}
\usetikzlibrary{matrix}
\usepackage{dsfont}
\usepackage{graphicx} 
\usepackage{stackrel}
\makeatletter

\@ifundefined{textcolor}{}
{%
 \definecolor{BLACK}{gray}{0}
 \definecolor{WHITE}{gray}{1}
 \definecolor{RED}{rgb}{1,0,0}
 \definecolor{GREEN}{rgb}{0,1,0}
 \definecolor{BLUE}{rgb}{0,0,1}
 \definecolor{CYAN}{cmyk}{1,0,0,0}
 \definecolor{MAGENTA}{cmyk}{0,1,0,0}
 \definecolor{YELLOW}{cmyk}{0,0,1,0}
}

\newenvironment{protocol*}[1]
  {
    \begin{center}
      \hrulefill\\
      \textbf{#1}
  }
  {
    \vspace{-1\baselineskip}
    \hrulefill
    \end{center}
  }

\newtheorem{thm}{Theorem}

\theoremstyle{definition}

\usepackage{xcolor}

\newtheorem{lemma}{Lemma}  
\def\bel{\begin{lemma}}
\def\eel{\end{lemma}}
\newtheorem{theorem}{Theorem}

\newtheorem*{proposition*}{Proposition}

\newtheorem{lem}[theorem]{Lemma}

\def\>{\rangle}
\def\<{\langle}
\makeatother

\begin{document}
\title{Quantum Nonlocality and Device-Independent Randomness Robust to Relaxations of Bell Assumptions}

\author{Ravishankar Ramanathan}
\email{ravi@cs.hku.hk}
\affiliation{Department of Computer Science, School of Computing and Data Science, The University of Hong Kong, Pokfulam Road, Hong Kong, China}
\author{Yuan Liu}
\email{yuan59@connect.hku.hk}
\affiliation{Department of Computer Science, School of Computing and Data Science, The University of Hong Kong, Pokfulam Road, Hong Kong, China}


\begin{abstract}
The question of certifying quantum nonlocality under a relaxation of the assumptions in the Bell theorem has gained traction, with potential for device-independent applications under weak seeds and cross-talk. Recently, it was shown that quantum nonlocality can be certified even under a simultaneous arbitrary (but not full) relaxation of the assumptions of Measurement Independence (MI) and Parameter Independence (PI), using states of local dimension $d = poly((1-\epsilon)^{-1})$ for an $\epsilon \in [0,1)$-relaxation. Here, we derive three results strengthening the state-of-art. Firstly, we show that states of constant local dimension $d$ are already sufficient to certify quantum nonlocality under arbitrary MI and PI relaxation, albeit in a non-robust manner. 
Secondly, and as a theoretical paradigm to derive the above, we introduce the notion of \textit{measurement-dependent parameter-dependent locality} as the set of input-output behaviors under simultaneous relaxations of measurement and parameter independence. We provide a rigorous characterisation of the vertices of the polytope of joint input-output behaviors that obey a $\mu$-relaxation of MI and $\epsilon$-relaxation of PI.  
We highlight a relation between nonlocality certification under PI relaxation and that under detection inefficiencies by pointing out alternative extremal correlations to the Eberhard correlations that also allow to achieve detection efficiency of $\eta = 2/3$ in the two-input scenario.
Finally, we study the implication of the relaxed assumptions for device-independent randomness certification. We analytically derive the quantum guessing probability for one player's outcomes in the CHSH Bell test, as a function of the noise in the test as well as of a leakage of an average amount of $I(X:B) < 1$ bits of input information per measurement round. 
\end{abstract}


\keywords{}

\maketitle

\textit{Introduction.-}
Bell nonlocality is a fundamental property of nature, whereby the outcomes of measurements on two or more spatially separated, entangled particles are more correlated than allowed in classical physics~\cite{Bell66,RMPBellnonlocality}. Besides its foundational interest, the phenomenon of Bell nonlocality has given rise to the possibility of Device-Independent (DI) quantum cryptography. In the DI framework, the honest users do not need to trust even the very devices used in the cryptographic protocol. They can instead directly verify the correctness and security by means of simple statistical tests of the device's input-output behavior, specifically checking for the violation of a Bell inequality. This violation would certify that the outputs of the device could not have been pre-determined, and as such could not be perfectly guessed by an adversary, allowing for the possibility of randomness extraction \cite{PAM+10, MS14}, amplification \cite{CR12, RBHH+16, RBH21, RHAP+18,BRGH+16}, and key distribution \cite{BHK05, DIQKD} besides self-testing \cite{Tsi93}, entanglement detection and quantification \cite{MBLHG13}.

The violation of Bell inequalities demonstrates that at least one reasonable physical property - such as outcome determinism, or no-signalling, or measurement independence - does not hold in nature. The question of the extent to which one or more of these properties must be relaxed is of both fundamental interest and of practical importance in understanding nonlocality as a resource in the aforementioned tasks \cite{Hall11}. 

In a bipartite Bell experiment such as in \cite{GMR+13, ZLR+22, WJSWZ98, HBD+15}, the inputs to their measurement devices are chosen by two parties Alice and Bob using local random number generators (RNGs). In an adversarial scenario, if the adversary could control or influence these RNGs, then it is possible to prepare classical devices that nevertheless appear nonlocal to the parties. The assumption that the measurement inputs are chosen independently of (any hidden variable describing the state of) the devices is referred to as measurement independence (MI) \cite{Bell77}. This assumption is related to the slightly stronger assumptions of freedom-of-choice or free will \cite{Hall10}. While various notions of relaxations of MI have been studied in the literature, a particularly relevant one was introduced in \cite{PRBLG14} building upon quantum protocols for randomness amplification of Santha-Vazirani (SV) sources \cite{CR12, RBHH+16} 
(the task whereby uniformly random bits are produced starting from a weak SV seed). Refs. \cite{PRBLG14, PG16} introduced a novel variant of Bell inequalities termed measurement dependent locality (MDL) inequalities on the joint input-output behavior in the Bell test. By systematically studying the convex polytope of measurement dependent local behaviors, it was shown that quantum nonlocality can be certified even under arbitrarily small amount of measurement independence \cite{PRBLG14}. The notion of MDL inequalities has since found widespread application in several directions \cite{TMWK25, SBB23, MA25}. 

The property of no-signalling, referred to as parameter independence (PI) at the level of a hidden variable (HV) model, is satisfied when the marginal distribution of Alice is independent of the setting chosen by Bob, and vice versa. The motivation for this assumption comes from relativistic causality \cite{HR19}, since altering the marginal distribution in one region via a change of measurement setting in a spacelike separated region would violate causality and lead to various paradoxes. Furthermore, since quantum theory respects no-signalling, it should be the case that even when the PI assumption is relaxed at the level of the HV model, averaging over hidden variables should lead to a no-signalling behavior at the observable level.


In this paper, we delineate the different notions of relaxations of the assumptions in the Bell theorem studied in the literature. In a recent result \cite{VRC25}, one of us showed how quantum nonlocality can be certified even under (a notion of) simultaneous relaxation of the assumptions of measurement independence and parameter independence by measurements on bipartite quantum systems of local dimension $d = poly((1-\epsilon)^{-1})$ for an $\epsilon \in [0,1)$-relaxation. Here, we first upgrade the result to show that a similar relaxation can be achieved already for constant $d$, specifically by suitable measurements on two-qubit systems (albeit in a much less robust manner).  Secondly and as a theoretical framework to show the above, we generalize the notion of Measurement Dependent Locality from \cite{PRBLG14, PG16} and introduce a novel paradigm of Measurement-Dependent Parameter-Dependent Locality (MDPDL). We show that the correlations under simultaneous relaxations of MI and PI form a convex polytope, and as such can be characterized using a finite set of Bell-like inequalities, allowing a systematic study of quantum nonlocality and its applications under weak seeds and leakage of input information. In the simplest $(2,2;2,2)$ Bell scenario, we derive certain facets of this convex polytope which thereby allow for optimal detection of nonlocality under this relaxation. For general Bell scenarios with binary outcomes, we prove a characterization of the extreme points of the MDPDL polytope, specifically showing that the probabilities $p_{ABXY}$ corresponding to extremal behaviors take values in a specific finite set that we derive. Our quantum correlations are alternative to the ones maximally violating the famous Eberhard inequality \cite{Eberhard93} in that they also allow a detection efficiency of $\eta = 2/3$ in the $(2,2;2,2)$ Bell scenario.
Finally, we apply our findings to the task of device-independent quantum randomness generation under leakage of input information. For the paradigmatic CHSH Bell test, we analytically compute the guessing probability of Alice's outcomes as a function of the observed Bell value and a leakage of $I(X:B)$ bits of information about Alice's input to Bob per run. We conclude with some interesting future directions and open questions.

\textit{Relaxing the assumptions in the Bell theorem.-}
In a Bell test, two separated parties, Alice and Bob, have access to two systems modeled as black-box devices. Over multiple runs in the test, Alice and Bob provide inputs to their device, denoted by respective random variables $X$ and $Y$ taking values $x \in \mathcal{X}$ and $y \in \mathcal{Y}$. They observe outputs denoted by respective random variables $A$ and $B$ taking values $a \in \mathcal{A}$ and $b \in \mathcal{B}$. The experiment is described by the set of probability distributions $\{p_{ABXY}\}$ termed a joint input-output behavior or correlation for the Bell scenario labeled as $\left(|\mathcal{A}|, |\mathcal{X}|; |\mathcal{B}|, |\mathcal{Y}|\right)$. 

Any underlying model of the correlation introduces an underlying variable $\Lambda$ on which the correlations depend. The probability $p_{ABXY}(abxy)$ can be written by Bayes theorem as
\begin{equation}
\label{eq:Bayesian-writing}
    p_{ABXY}(abxy)= \int d \lambda \; q_{\Lambda}(\lambda) \; p_{XY|\Lambda}(xy|\lambda) \; p_{AB|XY\Lambda}(ab|xy\lambda),
\end{equation}
for any $a,b,x,y$. When the underlying variable is discrete, the integration $\int d \lambda \; q_{\Lambda}(\lambda)$ is replaced by $\sum_{\lambda} q_{\Lambda}(\lambda)$.
An assumption termed \textit{outcome independence} (OI) states that given knowledge of $\lambda, x, y$, the measurement outcomes are uncorrelated, i.e., $p_{AB|XY\Lambda}(ab|xy\lambda) = p_{A|XY\Lambda}(a|xy\lambda) \cdot p_{B|XY\Lambda}(b|xy\lambda)$ for all $a,b,x,y,\lambda$. Any deterministic model is clearly outcome independent. The assumption of \textit{parameter independence} (PI) states that for every $\lambda$, the marginal distribution of outcomes seen by one party is independent of the input chosen by the other party, i.e., $p_{A|XY\Lambda}(a|xy\lambda) = p_{A|X\Lambda}(a|x\lambda)$ for all $a,x,y,\lambda$ and $p_{B|XY\Lambda}(b|xy\lambda) = p_{B|Y\Lambda}(b|y\lambda)$ for all $b,x,y,\lambda$. The assumption of \textit{measurement independence} (MI) states that measurement inputs can be chosen freely independent of the hidden variable, i.e., $p_{XY|\Lambda}(xy|\lambda) = p_{XY}(xy)$ for all $x,y,\lambda$. The conjunction $MI \wedge PI \wedge OI$ leads from \eqref{eq:Bayesian-writing} to the Local Hidden Variable (LHV) model
\begin{eqnarray}
\label{eq:LHV}
     &&p_{ABXY}(abxy)= \nonumber \\
     &&\int d \lambda \; q_{\Lambda}(\lambda) \; p_{XY}(xy) \; p_{A|X\Lambda}(a|x\lambda)\; p_{B|Y\Lambda}(b|y\lambda),
\end{eqnarray}
for all $a,b,x,y$. The relaxation of MI as mentioned in the Introduction is given for fixed $0 < l \leq h$ by $l \leq p(xy|\lambda) \leq h$ for all $x, y, \lambda$. The motivation comes from randomness amplification where inputs are chosen using bits taken from a $\mu$-SV source, which is defined by the property that each bit taken from the source is biased by at most $\mu \in [0, 1/2)$ from uniform, even when conditioned on all previous bits and other side information $\lambda$. 

The relaxation of PI requires more careful analysis. One way to define it is as the maximum possible shift in the marginal probability distribution for Bob's outcomes, induced by changing Alice's measurement setting, termed the \textit{average causal effect} \cite{CKBG15}
\begin{eqnarray}
    C_{Y \rightarrow A} &=& \sup_{x,a,y,y',\lambda} \big|p(a|x,do(y)\lambda) - p(a|x,do(y')\lambda)\big|, \nonumber \\
    C_{X \rightarrow B} &=& \sup_{y,b,x,x',\lambda} \big|p(b|do(x),y\lambda) - p(b|do(x'),y\lambda) \big|,
\end{eqnarray}
where $do(y)$ refers to resetting the input variable $Y$ to equal $y$. 
In \cite{VRC25} and in this paper, we consider a relaxation of PI in terms of the total variational distance as
\begin{eqnarray}
    \frac{1}{2} \sum_a | p_{A|XY\Lambda}(a|xy\lambda) - p_{A|XY\Lambda}(a|xy'\lambda) | &\leq& \epsilon_A\; \; \forall x,\lambda, y,y,'\nonumber \\
   \frac{1}{2} \sum_b | p_{B|XY\Lambda}(b|xy\lambda) - p_{B|XY\Lambda}(b|x'y\lambda) | &\leq& \epsilon_B \;\; \forall y,\lambda, x,x'.  \nonumber 
\end{eqnarray}
It is important to remark that in the $(2,2;2,2)$ Bell scenario of two binary variables per party that we focus on in this paper, the two relaxations are equivalent, i.e., in this scenario $C_{Y \rightarrow A} \leq \epsilon_A$, $C_{X \rightarrow B} \leq \epsilon_B$. So that a quantum correlation certifying nonlocality under arbitrary $\epsilon_A, \epsilon_B$ relaxation also does so under an arbitrary average causal effect of the input of one party upon the other party's output.

\textit{Measurement-dependent parameter-dependent locality.-} We say that the joint correlation (behavior) $p_{ABXY}$ is measurement-dependent and parameter dependent local (MDPDL) for given $l, h, \epsilon_A, \epsilon_B$ if it can be written in the form
\begin{eqnarray}
\label{eq:MDPDL}
    &&p_{ABXY}(abxy) = \nonumber \\ 
    &&\int d \lambda \; q_{\Lambda}(\lambda) \; p_{XY|\Lambda}(xy|\lambda) \; p_{A|XY\Lambda}(a|xy\lambda) \; p_{B|XY\Lambda}(b|xy\lambda), \nonumber
\end{eqnarray}
for all $a,b,x,y$ where 
\begin{eqnarray}
    &&l \leq p_{XY|\Lambda}(xy|\lambda) \leq h, \; \forall x,y,\lambda \nonumber \\
    &&\frac{1}{2} \sum_a | p_{A|XY\Lambda}(a|xy\lambda) - p_{A|XY\Lambda}(a|xy'\lambda) | \leq \epsilon_A\; \; \forall x,\lambda, y,y,'\nonumber \\
   &&\frac{1}{2} \sum_b | p_{B|XY\Lambda}(b|xy\lambda) - p_{B|XY\Lambda}(b|x'y\lambda) | \leq \epsilon_B \;\; \forall y,\lambda, x,x'.  \nonumber
\end{eqnarray}
The set of MDPDL behaviors $p_{ABXY}$ can be proven to be a convex polytope (see App. \ref{app:MDPDL-polytope}), and therefore can be characterized by its facets, that define strengthened Bell-like inequalities that we term MDPDL inequalities following \cite{PRBLG14, PG16}. A distribution cannot be explained by measurement dependence of $(l,h)$ and parameter dependence of $(\epsilon_A, \epsilon_B)$ if it violates one of the MDPDL inequalities.  

In App. \ref{app:MDPDL-polytope}, we provide a rigorous characterization of the extremal points of the MDPDL set for given $(l,h, \epsilon_A, \epsilon_B)$ in general Bell scenarios with binary outcomes, by proving that the corresponding probabilities $p_{ABXY}(abxy)$ only take values in a fixed finite set that we describe. The characterization of all extremal points allows to identify facet inequalities due to the duality between the vertex enumeration and facet enumeration problems. Software packages such as \textit{polymake} and \textit{cdd} exist for the purpose. Our vertex characterization thus allows to identify optimal facet inequalities for ruling out nonlocal hidden variable theories of specific parameter $(l,h, \epsilon_A, \epsilon_B)$ from given quantum experimental data. In App. \ref{app:facet_pd}, we demonstrate the above by deriving a facet of the polytope of conditional behaviors $p_{AB|XY}$ that obey $\epsilon$-PD constraints, where for simplicity we restrict to the case $\epsilon_A = \epsilon_B = \epsilon$. In App. \ref{app:MDPDL-ineq}, we show that the following inequality defines a supporting hyperplane of the $(l, h, \epsilon,\epsilon)$-MDPDL set:
\begin{widetext}
	\begin{equation}\label{eq:mdpd_ineq-main}
		l(1-\epsilon) \; \left[p_{ABXY}(0000)+\epsilon p_{ABXY}(1100) -\epsilon p_{XY}(00) \right]- h \; \left[p_{ABXY}(0101) + p_{ABXY}(1010) + p_{ABXY}(0011) \right] \leq 0.
\end{equation}    
\end{widetext}
We prove that there exist two-qubit correlations that lie outside the $(l, h, \epsilon,\epsilon)$-MDPDL for any $0 < l \leq h \leq 1$ and any $\epsilon \in [0,1)$. Specifically, suppose that Alice and Bob share the two-qubit state
\begin{equation}\label{eq:q-strategy-st-main}
        |\psi_{\theta}\rangle=\cos(\frac{\theta}{2})|00\rangle-\sin(\frac{\theta}{2})|11\rangle
    \end{equation}
    with a parameter $\theta=\arcsin{(3-\sqrt{4\epsilon+5})}$ dependent on the PI relaxation parameter $\epsilon$. 
Suppose that Alice and Bob measure, on this state, the respective observables $\{A_0, A_1\}$ and $\{B_0, B_1\}$ given as:
    \begin{equation}\label{eq:q_strategy-obs-main}
  \begin{split}
   A_0&=\frac{-(2+\sin\theta)\sqrt{1-\sin\theta}}{(2-\sin\theta)\sqrt{1+\sin\theta}} \sigma_z+\frac{-\sqrt2(\sin\theta)^{\frac{3}{2}}}{(2-\sin\theta)\sqrt{1+\sin\theta}} \sigma_x \nonumber, \\
  A_1&=-\sqrt{\frac{1-\sin\theta}{1+\sin\theta}}\sigma_z+\sqrt{\frac{2\sin\theta}{1+\sin\theta}}\sigma_x,
\end{split}
\end{equation}
with $B_0 = A_0$, $B_1 = A_1$ and where $\sigma_x$ and $\sigma_z$ are the usual Pauli operators. We show that the inequality \eqref{eq:mdpd_ineq-main} is violated by the resulting correlations. Thus, we find that quantum nonlocality is robust to arbitrary MI relaxation $0 \leq l < h$ and simultaneous arbitrary PI relaxation $0 \leq \epsilon < 1$.  

Interestingly by itself, when $\epsilon \rightarrow 1$, the quantum behavior yielded by the above qubit strategy can tolerate a detection efficiency of $2/3$, which is the same threshold achieved by the famous Eberhard correlations~\cite{Eberhard93}. In~\cite{github}, we show that the local fraction of this behavior, when the non-detected events are taken into account with detection efficiency close to $2/3$, is strictly less than $1$.


\textit{Analytical computation of the guessing probability in the CHSH Bell test under leakage of input information.-}
A vital advantage of device-independent (DI) protocols for quantum random number generation (DI-QRNG) and quantum key distribution (DI-QKD) is that they allow us to achieve security and classically monitor the performance of quantum devices irrespective of the noise, imperfections, lack of knowledge regarding their inner workings, or limited control over quantum devices. However, a major drawback of DI-QRNG is the expensive experimental requirements - one needs a Bell test with the detection loophole closed and with the measurement devices isolated from each other. It is clearly desirable to implement DI-QRNG using Bell tests on systems such as two ions in a ion trap coupled via their vibrational modes or with a single-photon quantum-dot based source feeding into a reconfigurable photonic chip \cite{FBM+24} as this allows to close the detection loophole and at the same time achieve high generation rates. Various approaches have been devised for this problem \cite{FBM+24, SPM13, FWEBC23, VEB+24}. 

Here, we highlight the application of our strengthened Bell inequalities to certify randomness even under leakage of input information from one measurement device to another. Our approach has the advantage of being implemented using two-qubit systems in the simplest Bell scenario, a feature which was not seen before \cite{VRC25}. Specifically, consider the $(2,2;2,2)$ scenario with two players performing two binary measurements each. Suppose that there is an average causal effect of $I(X:B) < 1$ bit, i.e., there is a one-way leakage of input information about the measurement on one ion that can influence the outcome of the measurement on a neighboring ion. We now proceed to analytically compute the guessing probability of Alice's outcomes by a quantum adversary as a function of the observed CHSH Bell value $\beta_{obs} \in [2,2\sqrt{2}]$ and $I(X:B)$.

The device-independent guessing probability $P^{(q)}_g(\beta_{obs}, \kappa, x^* = 0)$ for a quantum adversary for a fixed setting $x^*=0$ of Alice under the constraint of an observed value $\beta_{obs}$ and under communication parametrized by $\kappa$ is defined as (see \cite{RLWP25} for a derivation of the guessing probability):
\begin{eqnarray}
\label{eq:gpi}
   &&P^{(q)}_g(\beta_{obs}, \kappa, x^*=0) =\nonumber \\&& \max_{\{\textbf{p}^{a}\}, \{q_{a}\}} \sum_{a} q_{a} {p}^{a}_{AB|XY}(a,b|x^*=0,y) \nonumber \\
    &&\text{s.t.} \;\; \beta\left(\sum_{a} q_{a} {p}^{a}_{AB|XY}(a',b'|x',y')\right) = \beta_{obs}, \; \; \; \;  \nonumber \\
    && \; \; \; \;\; \textbf{p}^{a} \in \mathcal{Q}^{bip, \kappa}, \; \; \; q_{a} \geq 0 \; \; \forall a, \;\; \;  \sum_{a} q_{a} = 1.
\end{eqnarray}
Here $\mathcal{Q}^{bip, \kappa}$ denotes the set of bipartite quantum correlations with $\kappa$ amount of leakage, i.e., where Bob's observables depend upon Alice's inputs up to $\kappa$. Specifically, the CHSH inequality takes the form
\begin{equation}
 \beta(p_{AB|XY}) = \langle A_0 (B_{0,0} + B_{1,0}) + A_1 (B_{0,1} - B_{1,1}) \rangle \leq \beta_{c, \kappa},
\end{equation}
where $\beta_{c,\kappa} = 2(1+\kappa)$ is the classical bound of the expression. 
Here, $\| B_{y,x=0} - B_{y,x=1} \| \leq 2 \kappa$ for $y\in \{0,1\}$ is the constraint on the operator norm of the difference between Bob's observables for Alice's inputs $x \in \{0,1\}$. Here, as usual $\langle A_0 B_{0,0} \rangle = \sum_{a,b=0,1}(-1)^{a+b} p_{AB|XY}(a,b|x,y)$. In App. \ref{app:Pg-crosstalk} we derive a tight analytical expression for the guessing probability under leakage of input information for the CHSH Bell inequality, recovering the famous result for the CHSH guessing probability \cite{PAM+10} when $\kappa = 0$. Specifically, we prove that the quantum value of the inequality is given by
\begin{equation}\label{eq:sq}
    \beta_{q,\kappa}=\begin{cases}
        2\sqrt{2}(\kappa+\sqrt{1-\kappa^2}),\; & \text{ when } \kappa\in[0,\frac{1}{\sqrt{2}}];\\
        4, & \text{ when } \kappa\in[\frac{1}{\sqrt{2}},1).
    \end{cases}
\end{equation}
Defining
\begin{eqnarray}
&&\overline{P}_{g}^{(q)}(\beta_{obs},\kappa,x^*=0) - \frac{1}{2}= \nonumber \\
&&\begin{cases}
    \frac{1}{4}\sqrt{4-\big( \sqrt{\beta_{obs}^2-(2-4\kappa^2)^2} -4\kappa\sqrt{1-\kappa^2}\big)^2}, &  \kappa\in[0,\frac{1}{\sqrt{2}}];\\
    \frac{1}{4}\sqrt{\beta_{obs}(4-\beta_{obs})}, & \kappa\in[\frac{1}{\sqrt{2}},1),\nonumber\\
\end{cases}
\end{eqnarray} 
the guessing probability is given by the piecewise function:
\begin{widetext}
\begin{equation}\label{eq:pg_boud2}
    {P}_{g}^{(q)}(\beta_{obs},\kappa,x^*=0)= \begin{cases}
    \overline{P}_{g}^{(q)}(\beta_{obs},\kappa,x^*=0), & \text{ when } \beta_{obs}\in[\beta_{\kappa}^*,\beta_{q,\kappa}];\\
    \frac{\overline{P}_{g}^{(q)}(\beta_{\kappa}^*,\kappa,x^*=0)-1}{\beta_{\kappa}^*-\beta_{c,\kappa}} \beta_{obs}+\frac{\beta_{\kappa}^*-\beta_{c,\kappa}\overline{P}_{g}^{(q)}(\beta_{\kappa}^*,\kappa,x^*=0)}{\beta_{\kappa}^*-\beta_{c,\kappa}}, & \text{ when } \beta_{obs} \in[\beta_{c,\kappa},\beta_{\kappa}^*].\\
\end{cases}
\end{equation}
\end{widetext}
Here, $\beta_{\kappa}^*$ is found as the maximizer of the following optimization:
\begin{equation}
\begin{split}
	 &\max_{\beta \in[\beta_{c,\kappa}, \beta_{q,\kappa}]} \quad \frac{\overline{P}_{g}^{(q)}(\beta,\kappa,x^*=0) - 1}{\beta- \beta_{c,\kappa}}
\end{split}
\end{equation}
That is, $\beta_{\kappa}^*$ is the unique real solution in $[\beta_{c,\kappa},\beta_{q,\kappa}]$ of $\frac{\partial \overline{P}_{g}^{(q)}(\beta,\kappa,x^*=0)}{\partial \beta}\left(\beta-\beta_{c, \kappa}\right)=\overline{P}_{g}^{(q)}(\beta,\kappa,x^*=0)-1$.

While interesting in its own right for the analysis of CHSH-based DIQRNG experiments without the enforcement of spacelike separation between devices, it is of more interest to derive the guessing probability for the MDPDL inequality \eqref{eq:mdpd_ineq-main} under two-sided leakage of input information, a task which we pursue in forthcoming work. 

\textit{Quantum correlations that cannot be simulated with $I(X:B)=1-\delta$ bits of communication, for $\delta > 0$.-}
There has been much interest in the amount of classical communication that must be added to the LHV model to simulate a given quantum correlation. In particular, following on from a seminal result by Bacon and Toner \cite{BT03}, it has been proven that all correlations resulting from projective measurements on the two-qubit maximally entangled state can be simulated with LHV augmented by just one bit of classical communication. Sidayaja et al.'s \cite{SLYS23} numerical study using neural networks has gathered strong evidence that correlations resulting from projective measurements on any two-qubit state can be simulated with one bit of communication. A comprehensive work by Marton et al \cite{MBDV24} has identified minimal Bell scenarios where quantum nonlocality beats one bit of communication. In these results, the model considered is of the form $\int d\lambda \; q(\lambda) \; p(a|x\lambda) \; p(b|yc\lambda)$ where Bob's marginal also depends on a communicated bit $c = c(x,\lambda) \in \{0,1\}$. 

The PI relaxation considered in this paper is naturally related to the above investigations of the communication cost of nonlocality. Specifically, the result in the previous section translates to a family of correlations resulting from projective measurements on weakly entangled two-qubit states that require an amount of communication $I(X:B) = 1 - \delta$ to simulate classically, for arbitrary $\delta > 0$. We thus show that if the conjecture holds that all correlations from projective measurements on two-qubit states can be simulated with $1$ bit communication, then this $1$ bit is necessary as any small deviation from this value on average would not be sufficient to simulate the correlations.

\textit{Conclusions and Open Questions.-} In this paper, we have seen that quantum nonlocality can be certified even when the class of hidden variable models is relaxed (for arbitrary simultaneous relaxations of the MI and PI assumptions behind these models). A natural question is to generalise the relaxation to a non-IID (non-independent and identically distributed) setting. Specifically, in \cite{VRC25}, following \cite{SBB23} we have shown how quantum nonlocality can be certified even for a non-IID MI and IID PI relaxation. An important first question for forthcoming research would be to see if the PI relaxation can be similarly relaxed in a non-IID manner while still certifying nonlocality. Technically, this amounts to the condition that in a fraction of the runs, an $\epsilon$-relaxation of PI is assumed, while in the remaining runs, full dependence of either player's outcome on the other player's input is allowed. A second important question would be to build upon the analytical derivation of the single-round min-entropy given here and formulate a security proof for DI randomness expansion and amplification in the scenario of leakage of input information. A third important question would be to tackle a simultaneous relaxation also of the outcome independence assumption, either in the form considered in \cite{VRC25} or as a simultaneous relaxation of the causal outcome independence \cite{RGCC+16} assumption in the causal model picture. Finally, we have found that the two-qubit point that achieves arbitrary relaxation of MI and PI is also able to tolerate detection efficiency close to $2/3$ which is the threshold achieved by Eberhard correlation \cite{Eberhard93} in this $(2,2;2,2)$ scenario. We explore the relationship between our relaxations of the Bell theorem assumptions and the detection efficiency required to certify nonlocality in forthcoming work.

\textit{Acknowledgments.-}
We acknowledge support from the General Research Fund (GRF) Grant No.\ 17211122, and the Research Impact Fund (RIF) Grant No.\ R7035-21.

\bibliographystyle{apsrev4-2}

\onecolumngrid
\appendix
\newpage

\section{Characterization of the Measurement-Dependent Parameter-Dependent Local Set} 
\label{app:MDPDL-polytope}
We consider a bipartite Bell scenario in which an observer (Alice) in her laboratory chooses to measure $x \in \mathcal{X}$ and obtains an outcome $a \in \mathcal{A}$. A second observer (Bob) in his laboratory chooses to measure $y \in \mathcal{Y}$ and obtains $b \in \mathcal{B}$. Alice and Bob together observer the joint probability $p_{ABXY}(abxy)$ of outcomes $(a,b)$ and inputs $(x,y)$. The set of joint probabilities $\{p_{ABXY}(abxy)\}$ is called a joint input-output correlation or behavior for the Bell scenario $\left(|\mathcal{X}|, |\mathcal{A}|; |\mathcal{Y}|, |\mathcal{B}| \right)$. Any such behavior can be represented as
\begin{equation}
\label{eq:pjoint-gen}
    p_{ABXY}(abxy) = \int d\lambda \;q_{\Lambda}(\lambda) \; p_{XY|\Lambda}(xy|\lambda) \; p_{AB|XY\Lambda}(ab|xy\lambda), \;\; \forall a,b,x,y,
\end{equation}
where $\{\lambda\}$ denotes a set of hidden variables, possible common pasts that Alice and Bob's systems share. Here, to be precise, we have written the random variable as uppercase letters in the subscript and the values taken by the random variable as lowercase letters in the argument. 
The assumption of Outcome Independence (OI) which we make throughout this paper states that
\begin{equation}
p_{AB|XY\Lambda}(ab|xy\lambda) = p_{A|XY\Lambda}(a|xy\lambda) \; p_{B|XY\Lambda}(b|xy\lambda) \;\; \forall a,b,x,y,\lambda.
\end{equation}
In other words, any correlation between the outcomes of both parties can be traced back to the hidden variable and inputs, so that conditioning upon these variables results in a product distribution. 
Under OI, the joint behavior \eqref{eq:pjoint-gen} reduces to
\begin{equation}
\label{eq:pjoint-OI}
    p_{ABXY}(abxy) = \int d\lambda \; q_{\Lambda}(\lambda) \; p_{XY|\Lambda}(xy|\lambda) \; p_{A|XY\Lambda}(a|xy\lambda) \; p_{B|XY\Lambda}(b|xy\lambda), \;\; \forall a,b,x,y.
\end{equation}
We are interested in the set of correlations that obey Outcome Independence but only a relaxed form of Measurement and Parameter Independence. As mentioned in the text, we impose 
\begin{eqnarray}
\label{eq:relaxations-MIPI}
    &&l \leq p_{XY|\Lambda}(xy|\lambda) \leq h, \; \forall x,y,\lambda \nonumber \\
    &&\frac{1}{2} \sum_a \big| p_{A|XY\Lambda}(a|xy\lambda) - p_{A|XY\Lambda}(a|xy'\lambda) \big| \leq \epsilon_A,\; \; \forall x,\lambda, y,y,'\nonumber \\
   &&\frac{1}{2} \sum_b \big| p_{B|XY\Lambda}(b|xy\lambda) - p_{B|XY\Lambda}(b|x'y\lambda) \big| \leq \epsilon_B, \;\; \forall y,\lambda, x,x'.  
\end{eqnarray}
We refer to the set of joint input-output behaviors of the form in \eqref{eq:pjoint-OI} that obey the constraints \eqref{eq:relaxations-MIPI} as $(l,h,\epsilon_A, \epsilon_B)$-Measurement-dependent Parameter-dependent local (MDPDL) behaviors. 

Before we proceed to analyse the set of MDPDL behaviors, it is worth noting that there are multiple ways to relax the assumptions of measurement and parameter independence. The relaxation of MI that we follow here stems from a practical assumption in \cite{PRBLG14, PG16} that the players choose their measurement settings using a weak seed that corresponds to a Santha-Vazirani source \cite{CR12}. Specifically, even when conditioned on the past $\lambda$ (and on any previous bits), each bit produced by the source has some randomness so that $p(xy|\lambda) \geq l > 0$ for all $x, y, \lambda$. We will say then that MI is weakened arbitrarily when quantum nonlocality is observed for any $l > 0$. We note that most of our considerations in this paper will be in the $(2,2;2,2)$ Bell scenario so that the statement is justified. For if $p(xy|\lambda)$ is allowed to be $0$ for some input pair $(x,y)$ in this scenario, then the quantum (and even general no-signalling) value for any Bell expression can be simulated with a local hidden variable model in this scenario, so that the value achieved by any quantum correlation can be reproduced with such $\lambda$ and no nonlocality can be expected. The relaxation of PI can also be considered in different ways. One way to define it is as the maximum possible shift in the marginal probability distribution for Bob's outcomes, induced by changing Alice's measurement setting, termed the \textit{average causal effect} (such relaxation was considered for example in \cite{CKBG15, Hall11}):
\begin{eqnarray}
    C_{Y \rightarrow A} &=& \sup_{x,a,y,y',\lambda} \big|p(a|x,do(y)\lambda) - p(a|x,do(y')\lambda)\big|, \nonumber \\
    C_{X \rightarrow B} &=& \sup_{y,b,x,x',\lambda} \big|p(b|do(x),y\lambda) - p(b|do(x'),y\lambda) \big|,
\end{eqnarray}
where $do(y)$ refers to resetting the input variable $Y$ to equal $y$. 
In \cite{VRC25} and in this paper, we consider a relaxation of PI in terms of the total variational distance (in \eqref{eq:relaxations-MIPI}). 
It is important to remark that in the $(2,2;2,2)$ Bell scenario of two binary variables per party that we focus on in this paper, the two relaxations are equivalent, i.e., in this scenario $C_{Y \rightarrow A} \leq \epsilon_A$, $C_{X \rightarrow B} \leq \epsilon_B$. So that a quantum correlation certifying nonlocality under arbitrary $\epsilon_A, \epsilon_B$ relaxation also does so under an arbitrary average causal effect of the input of one party upon the other party's output.

In this Appendix, we are interested in the $(l,h,\epsilon_A,\epsilon_B)$-MDPDL set of joint input output behaviors $p_{ABXY}$ in Bell scenarios with two oucomes per player, but an arbitrary number of inputs. Specifically, we show that the set is a convex polytope and then derive a characterisation of the extreme points (vertices) of the polytope. 

Let us first consider the input set $\mathcal{P}^{XY,(l,h)}_1$ defined as
\begin{equation}
    \mathcal{P}^{XY,(l,h)}_1 = \big\{ p_{XY}: l \leq p_{XY}(xy) \leq h \; \forall (x, y) \in \mathcal{X} \times \mathcal{Y} \big\}.
\end{equation}
It is clear that this set is a convex polytope, as the constraints are all linear. In \cite{PRBLG14, PG16}, the vertices of this set were derived and it was found that for the case that $l < h$, the vertices are of the form
\begin{equation}
\label{eq:vertices-inputpolytope}
    V_{\mathcal{P}^{XY,(l,h)}_1} = \big\{ p_{XY}: \; \exists \pi \in S_{|\mathcal{X}||\mathcal{Y}|} \; \text{s.t.} \; p_{XY} = \pi(A^{(l,h)})\big\},
\end{equation}
where the set $A^{(l,h)}$ is defined as
\begin{equation}
    A^{(l,h)} := \big\{ \underbrace{h,\ldots, h}_{n \; \text{times}}, \underbrace{l,\ldots,l}_{|\mathcal{X}||\mathcal{Y}|-n-1 \; \text{times}}, 1- n h - (|\mathcal{X}||\mathcal{Y}|-n-1)l \big\},
\end{equation}
where $n = \big\lfloor \frac{1-|\mathcal{X}||\mathcal{Y}|l}{h-l}\big\rfloor$. Here $\pi$ is any permutation chosen from the set $S_{|\mathcal{X}||\mathcal{Y}|}$ of permutations of $|\mathcal{X}||\mathcal{Y}|$ elements. On the other hand, when $l=h$, the input set $\mathcal{P}_1^{(l,h)}$ has only one extreme point with entries $p_{XY}(x,y)= \frac{1}{|\mathcal{X}||\mathcal{Y}|}$ for all $x,y$. 

We now move to consider the set of conditional output distributions $p_{AB|XY}$ obeying outcome independence and the $(\epsilon_A, \epsilon_B)$ relaxations of parameter independence. That is, we consider the set $\mathcal{P}^{AB,(\epsilon_A, \epsilon_B)}_2$ defined as
\begin{eqnarray}\label{eq:P2-PD-polytope}
    \mathcal{P}^{AB,(\epsilon_A,\epsilon_B)}_2 = \bigg\{ p_{AB|XY} : p_{AB|XY}(ab|xy) =&& \int d\lambda \; q_{\Lambda}(\lambda) \;  p_{A|XY\Lambda}(a|xy\lambda) \; p_{B|XY\Lambda}(b|xy\lambda), \\ &&q_{\Lambda}(\lambda) \geq 0 \; \forall \lambda, \; \; \int d \lambda \; q_{\Lambda}(\lambda) = 1, \nonumber \\
    &&p_{A|XY\Lambda}(a|xy\lambda) \geq 0 \;  \forall a, x, y, \lambda, \; p_{B|XY\Lambda}(b|xy\lambda) \geq 0 \; \forall b,x,y,\lambda, \\&& \sum_{a}  p_{A|XY\Lambda}(a|xy\lambda) = 1 \; \; \forall x, y, \lambda, \; \; \sum_b p_{B|XY}\Lambda(b|xy\lambda)=1 \; \forall x,y,\lambda, \nonumber \\
    &&  \frac{1}{2} \sum_a \bigg| p_{A|XY\Lambda}(a|xy\lambda) - p_{A|XY\Lambda}(a|xy' \lambda) \bigg| \leq \epsilon_A \; \; \; \forall x, y, y', \lambda, \nonumber \\
    &&  \frac{1}{2} \sum_b \bigg| p_{B|XY\Lambda}(b|xy\lambda) - p_{B|XY\Lambda}(b|x'y \lambda) \bigg| \leq \epsilon_B \; \; \; \forall x, x', y, \lambda\bigg\}.
\end{eqnarray}
As we shall see, this set is a polytope since despite the product nature of the individual probabilities. Furthermore, despite the apparent nonlinearity induced by the absolute value constraints, we observe that each absolute value term can be replaced by a linear term with a choice of a $\pm$ sign. That is, each absolute value constraint of the $\epsilon_A$ type can be converted to $2^{|A|}$ linear constraints, and similarly for constraints of $\epsilon_B$ type. 
We now proceed to characterise the vertices of the polytope $\mathcal{P}^{AB,(\epsilon_A, \epsilon_B)}_2$ via the following theorem. 



\begin{thm}
    Consider the Bell scenario $(|\mathcal{X}|,|\mathcal{A}|; |\mathcal{Y}|,|\mathcal{B}|)$ for arbitrary $|\mathcal{X}|, |\mathcal{Y}| \geq 2$, with $|\mathcal{A}|= |\mathcal{B}| = 2$. For any extreme point (vertex) $p^{ext}_{AB|XY}$ of the set $\mathcal{P}^{AB,(\epsilon_A, \epsilon_B)}_2$ given in \eqref{eq:P2-PD-polytope}, it holds that $p^{ext}_{AB|XY}(ab|xy) \in \big\{0, \epsilon_A, 1 - \epsilon_A, 1\} \times \{0, \epsilon_B, 1-\epsilon_B, 1\}$ for all $a,b,x,y$. 
\end{thm}
\begin{proof}
Let $\mathcal{P}^{A, \epsilon_A}$ denote the polytope of marginal correlations $\{p_{A|XY}\}$ defined as follows:
\begin{eqnarray}
\label{eq:polytope-A}
    \mathcal{P}_2^{A, \epsilon_A} = \bigg\{ p_{A|XY} : p_{A|XY}(a|x,y) &=& \int d \lambda \;  q_{\Lambda}(\lambda) \;  p_{A|XY\Lambda}(a|xy\lambda), \; \; q_{\Lambda}(\lambda) \geq 0 \; \forall \lambda, \; \; \int d\lambda \; q_{\Lambda}(\lambda) = 1, \nonumber \\
    &&p_{A|XY\Lambda}(a|xy\lambda) \geq 0 \; \; \forall a, x, y, \lambda, \; \; \sum_a  p_{A|XY\Lambda}(a|xy\lambda) = 1 \; \; \forall x, y, \lambda, \nonumber \\
    &&  \frac{1}{2} \sum_a \bigg| p_{A|XY\Lambda}(a|xy\lambda) - p_{A|XY\Lambda}(a|xy' \lambda) \bigg| \leq \epsilon_A \; \; \; \forall x, y, y', \lambda \bigg\}.
\end{eqnarray}
Again, the reason why $\mathcal{P}_2^{A, \epsilon_A}$ is a convex polytope despite the apparently nonlinear constraint is that each of the absolute value constraints in the $\epsilon_A$ parameter dependent relaxation is equivalent to $2^{|\mathcal{A}|}$ linear constraints obtained by choosing a $+/-$ sign for each $| \cdot |$ term. 

Let $\mathcal{P}_2^{B, \epsilon_B}$ denote the analogous polytope of marginal correlations $\{p_{B|XY} \}$ for Bob, i.e., 
\begin{eqnarray}
\label{eq:subpolytope-B}
    \mathcal{P}_2^{B, \epsilon_B} = \bigg\{ p_{B|XY} : p_{B|XY}(b|xy) &=& \int d\lambda \; q_{\Lambda}(\lambda) \;  p_{B|XY\Lambda}(b|xy\lambda), \; \; q_{\Lambda}(\lambda) \geq 0 \; \forall \lambda, \; \; \int d\lambda\; q_{\Lambda}(\lambda) = 1, \nonumber \\
    &&p_{B|XY\Lambda}(b|xy\lambda) \geq 0 \; \; \forall b, x, y, \lambda, \; \; \sum_b  p_{B|XY\Lambda}(b|xy\lambda) = 1 \; \; \forall x, y, \lambda, \nonumber \\
    &&  \frac{1}{2} \sum_b \bigg| p_{B|XY\Lambda}(b|xy\lambda) - p_{B|XY\Lambda}(b|x'y \lambda) \bigg| \leq \epsilon_B \; \; \; \forall y, x, x', \lambda \bigg\}.
\end{eqnarray}

We will also have occasion to consider the polytopes for fixed inputs, $\mathcal{P}_2^{x^*,A, \epsilon_A}$ and $\mathcal{P}_2^{y^*, B, \epsilon_B}$, defined as follows. 
\begin{eqnarray}
\label{eq:subpolytope-A-fixedx}
    \mathcal{P}^{x^*, A, \epsilon_A} = \bigg\{ p_{A|X=x^*,Y} : p_{A|XY}(a|x^*,y) &=& \int d \lambda \;  q_{\Lambda}(\lambda) \;  p_{A|XY\Lambda}(a|x^*y\lambda), \; \; q_{\Lambda}(\lambda) \geq 0 \; \forall \lambda, \; \; \int d\lambda \; q_{\Lambda}(\lambda) = 1, \nonumber \\
    &&p_{A|XY\Lambda}(a|x^*y\lambda) \geq 0 \; \; \forall a, y, \lambda, \; \; \sum_a  p_{A|XY\Lambda}(a|x^*y\lambda) = 1 \; \; \forall y, \lambda, \nonumber \\
    &&  \frac{1}{2} \sum_a \bigg| p_{A|XY\Lambda}(a|x^*y\lambda) - p_{A|XY\Lambda}(a|x^*y' \lambda) \bigg| \leq \epsilon_A \; \; \; \forall y, y', \lambda \bigg\}.
\end{eqnarray}
\begin{eqnarray}
\label{eq:subpolytope-B-fixedy}
    \mathcal{P}_2^{y^*,B, \epsilon_B} = \bigg\{ p_{B|X,Y=y^*} : p_{B|XY}(b|xy^*) &=& \int d\lambda \; q_{\Lambda}(\lambda) \;  p_{B|XY\Lambda}(b|xy^*\lambda), \; \; q_{\Lambda}(\lambda) \geq 0 \; \forall \lambda, \; \; \int d\lambda\; q_{\Lambda}(\lambda) = 1, \nonumber \\
    &&p_{B|XY\Lambda}(b|xy^*\lambda) \geq 0 \; \; \forall b, x, \lambda, \; \; \sum_b  p_{B|XY\Lambda}(b|xy^*\lambda) = 1 \; \; \forall x, \lambda, \nonumber \\
    &&  \frac{1}{2} \sum_b \bigg| p_{B|XY\Lambda}(b|xy^*\lambda) - p_{B|XY\Lambda}(b|x'y^* \lambda) \bigg| \leq \epsilon_B \; \; \; \forall x, x', \lambda \bigg\}.
\end{eqnarray}
\begin{lem}
\label{lem:Alice-polytope-fixedx}
Let $p^{ext}_{A|X=x^*,Y}$ be an extremal point (a vertex) of the polytope $\mathcal{P}_2^{x^*,A, \epsilon_A}$ for arbitrary $|\mathcal{X}|, |\mathcal{Y}| \geq 2$ and with $|\mathcal{A}|=2$. It holds that $p^{ext}_{A|XY}(a|x^*y) \in \{0, \epsilon_A, 1- \epsilon_A, 1\}$ for all $a, y$. 
\end{lem}
\begin{proof}
Let us write the polytope $\mathcal{P}_2^{x^*,A, \epsilon_A}$ in standard form as $\texttt{A}.|\texttt{p} \rangle \leq |\texttt{b}\rangle$. Here, the first $|\mathcal{Y}||\mathcal{A}|$ rows of $\texttt{A}$ encode the non-negativity constraints on the probabilities $p_{A|X=x^*,Y}$, the next $2|\mathcal{Y}|$ rows encode the normalization constraint for each value of $y$, and the next $\binom{|\mathcal{Y}|}{2} \cdot 2^{|\mathcal{A}|}$ rows encode the $\epsilon_A$ relaxations of the PI condition. For instance, for the case $|\mathcal{Y}| = |\mathcal{A}| = 2$, the form of the polytope is as follows:
\[
\begin{matrix}
\begin{aligned}
|\mathcal{Y}||\mathcal{A}| \; \text{rows}
  &\left\{\begin{matrix}
    \\
    \\
     \\
     \\
     \\
  \end{matrix}\right.  %
  &\begin{matrix}
  \\\phantom{\cdots}
  \end{matrix}\\ %
2|\mathcal{Y}| \; \text{rows}
  &\left\{\begin{matrix}
    \\
    \\
     \\
     \\
     \\
     \\
  \end{matrix}\right.  %
  &\begin{matrix}
  \\\phantom{\cdots}
  \end{matrix}\\ %
\binom{|\mathcal{Y}|}{2} \cdot 2^{|\mathcal{A}|}  \; \text{rows}
  &\left\{\begin{matrix}
    \\
    \\
     \\
     \\
     \\
     \\
  \end{matrix}\right.  %
  &\begin{matrix}
  \\\phantom{\cdots}
  \end{matrix}\\ %

 \end{aligned}
 \end{matrix}
 \begin{bmatrix}
 -1  &  &
  &   \\[0.5em]
   & -1 & 
 &   \\[0.5em]
  &  & -1 &
   \\[0.5em]
  &  &  &
 -1   \\[0.5em]
\hline
 1 & 
 1 & 0 & 0\\[0.5em]
 -1  &
 -1 & 0  &  0 \\[0.5em]
0  & 0& 
 1  & 1 \\[0.5em]
  0  &
  0 & -1  & -1 \\[0.5em]
\hline  
\frac{1}{2} & \frac{1}{2} &  \frac{-1}{2} & \frac{-1}{2} \\[0.5em]
\frac{1}{2} & \frac{-1}{2}  & \frac{-1}{2} &  \frac{1}{2} \\[0.5em]
\frac{-1}{2} & \frac{1}{2} & \frac{1}{2} & \frac{-1}{2}  \\[0.5em]
\frac{-1}{2} & \frac{-1}{2} & \frac{1}{2} & \frac{1}{2}  \\[0.5em]
  \end{bmatrix} \cdot
\begin{bmatrix}
\begin{aligned}
p_{A|XY}(a=0|x^*,y=0) \\[0.5em]
p_{A|XY}(a=1|x^*,y=0) \\[0.5em]
p_{A|XY}(a=0|x^*,y=1) \\[0.5em]
p_{A|XY}(a=1|x^*,y=1) \\[0.5em]
\end{aligned}
\end{bmatrix}  
\leq  
\begin{bmatrix}
\begin{aligned}
0 \\[0.5em]
0 \\[0.5em]
0 \\[0.5em]
0 \\[0.5em]
1 \\[0.5em]
-1 \\[0.5em]
1 \\[0.5em]
-1 \\[0.5em]
\epsilon_A \\[0.5em]
\epsilon_A \\[0.5em]
\epsilon_A \\[0.5em]
\epsilon_A \\[0.5em]
\end{aligned}
\end{bmatrix}.  
  \]

The dimension of the polytope $\mathcal{P}_2^{x^*,A, \epsilon_A}$ is $d = |\mathcal{Y}|\left(|\mathcal{A}| - 1 \right)$.
This comes from the normalisation constraint which imposes that there are $|\mathcal{A}|-1$ free parameters for each value of $y \in \mathcal{Y}$. By writing $p_{A|XY}(a=|\mathcal{A}|-1 \big|x^*y) = 1 - \sum_{a=0}^{|\mathcal{A}|-2} p_{A|XY}(a|x^*y)$, we parametrise the box with $|\mathcal{A}|$ entries $p_{A|XY}(a|x^*y)$ for $a=0$ to $a=|\mathcal{A}|-2$. The polytope can then be written as $\texttt{A}'.|\texttt{p}'\rangle \leq |\texttt{b}'\rangle$, with the first $d=|\mathcal{Y}|(|\mathcal{A}|-1)$ rows of $\texttt{A}'$ encoding the non-negativity constraints, and the remaining $\binom{|\mathcal{Y}|}{2}\cdot 2^{|\mathcal{A}|-1}$ rows encoding the $\epsilon_A$ relaxations of the PI condition. 

Now, recall that point in the polytope is a vertex (an extreme point or a basic feasible solution) if it cannot be expressed as a convex combination of two other distinct points in the polytope. This geometric condition is also equivalent to an algebraic condition: at a vertex, the number of saturated (also known as active, or binding) constraints must be at least the dimension $d$. If we consider only the non-redundant inequalities, exactly $d$ linearly independent inequalities are saturated.

From hereon, we restrict to the case of interest, namely $|\mathcal{A}|=2$. We prove that all entries in a vertex only take values in $\{0,\epsilon_A,1-\epsilon_A,1\}$ by an induction on number of inputs $y$.

Consider the base case: $|\mathcal{A}|=2, \; |\mathcal{Y}|=2$. In this case, we have  $|\texttt{p}'\rangle = \left[p_{A|XY}(a=0|x^*=0,y=0), \; p_{A|XY}(a=0|x^*=0,y=1) \right]^T$, and the dimension is $d=2$.  Explicitly, we have
\[
 \begin{bmatrix}
 -1  &   \\[0.5em]
   & -1 \\[0.5em]
\hline
1 & -1 \\[0.5em]
-1 & 1 \\[0.5em]
\end{bmatrix} \cdot
\begin{bmatrix}
\begin{aligned}
p_{A|XY}(a=0|x^*,y=0) \\[0.5em]
p_{A|XY}(a=0|x^*,y=1) \\[0.5em]
\end{aligned}
\end{bmatrix}  
\leq  
\begin{bmatrix}
\begin{aligned}
0 \\[0.5em]
0 \\[0.5em]
\epsilon_A \\[0.5em]
\epsilon_A \\[0.5em]
\end{aligned}
\end{bmatrix}.  
  \]
For a point $|\texttt{p}'\rangle$ to represent a vertex, it must saturate two of the non-negativity and PI relaxation conditions. There are three possibilities: (i) Both of the non-negativity constraints are active. In this case, we have $p_{A|XY}(a=0|x^*,y=0) = p_{A|XY}(a=0|x^*,y=1) = 0$ with the corresponding $p_{A|XY}(a=1|x^*,y=0) = p_{A|XY}(a=1|x^*,y=1) = 1$. (ii) One non-negativity constraint and a PI relaxation condition is active. In this case, one of the probabilities in $|\texttt{p}'\rangle$ is $0$ with the other being equal to $\epsilon_A$. Correspondingly, the probabilities for $a=1$ take values $1$ and $1 - \epsilon_A$ respectively. (iii) None of the non-negativity constraints is active. In this case, the point cannot be a vertex as the PI relaxation submatrix of $\texttt{A}'$ only has rank one. In all three cases, we see that the entries in the vertex only take values in $\{0,\epsilon_A,1-\epsilon_A, 1\}$. 

We now proceed to perform an induction on the number of inputs $y$. Suppose that for $|\mathcal{Y}|=l$ that the statement holds. Explicitly, we have
\[
\begin{matrix}
\begin{aligned}
|\mathcal{Y}| \; \text{rows}
  &\left\{\begin{matrix}
    \\
    \\
     \\
     \\
     \\
     \\
  \end{matrix}\right.  %
  &\begin{matrix}
  \\\phantom{\cdots}
  \end{matrix}\\ |\mathcal{Y}|\left(|\mathcal{Y}|-1 \right)   \; \text{rows}
  &\left\{\begin{matrix}
    \\
    \\
     \\
     \\
     \\
     \\
     \\
     \\
     \\\
     \\
     \\
     \\
     \\
     \\
     \\
     \\
  \end{matrix}\right.  %
  &\begin{matrix}
  \\\phantom{\cdots}
  \end{matrix}\\ %

 \end{aligned}
 \end{matrix}
\begin{bmatrix}
-1  & & & &   \\[0.5em]
   & -1 & & &  \\[0.5em]
   & & & \ddots & \\[0.5em]
  & & & &  -1 \\[0.5em]
\hline
1 & -1 &  &  & \\[0.5em]
1 &  & -1  &  &  \\[0.5em]
&&&\ddots \\[0.5em]
%
1 &  &  &  & -1 \\[0.5em]
\hline
 & 1 & -1 &  &  \\[0.5em]
&&&\ddots \\[0.5em]
%
 & 1 &  &  & -1  \\[0.5em]
\hline
 &  & 1 & -1 &   \\[0.5em]
 &  & 1 &  & -1 \\[0.5em]
\hline
 &  &  & 1 & -1 \\[0.5em]
\vdots \; \\[0.5em]
\end{bmatrix} \cdot
\begin{bmatrix}
\begin{aligned}
&p_{A|XY}(a=0|x^*,y=0) \\[0.5em]
&p_{A|XY}(a=0|x^*,y=1) \\[0.5em]
&\qquad \qquad \quad \vdots \\[0.5em]
&p_{A|XY}(a=0|x^*,y=l-1)\\
\end{aligned}
\end{bmatrix}  
\leq  
\begin{bmatrix}
\begin{aligned}
0 \\[0.5em]
\vdots\; \\[0.5em]
\vdots\; \\[0.5em]
0 \\[0.5em]
\epsilon_A \\[0.5em]
\vdots \; \\[0.5em]
\vdots\; \\[0.5em]
\vdots\; \\[0.5em]
\vdots\; \\[0.5em]
\vdots\; \\[0.5em]
\vdots\; \\[0.5em]
\epsilon_A \\[0.5em]
\end{aligned}
\end{bmatrix},
\]
Here, the dots at the end of $\texttt{A}'$ indicate additional rows that are sign flipped versions of the previous rows. 
We prove that the statement then holds for $|\mathcal{Y}|=l+1$. Firstly, observe that any vertex $|\texttt{p}'\rangle$ for $l+1$ inputs of Bob must also be a vertex for the sub-behavior of $l$ inputs. For, if a convex decomposition into points in the polytope exists for the sub-behavior, then such a convex decomposition can be extended into one for the whole behavior $|\texttt{p}'\rangle$ which would be a contradiction to the assumption that $|\texttt{p}'\rangle$ is a vertex. Now, in going from $l$ to $l+1$ inputs, we add one term $p_{A|XY}(a=0|x^*,y=l+1)$ to the behavior. And correspondingly the dimension increases by one from $l$ to $(l+1)$, so that one additional linearly independent constraint needs to be active in the vertex. The rank of the submatrix of $\texttt{A}'$ corresponding to the PI relaxation conditions increases by one from $(l-1)$ to $l$. There are two possibilities: (i) The new non-negativity constraint is active. In this case, we have $p_{A|XY}(a=0|x^*,y=l) = 0$, with the corresponding $p_{A|XY}(a=1|x^*,y=l)=1$, and the required condition is satisfied. (ii) A new PI relaxation condition is active. In this case, we have without loss of generality that $p_{A|XY}(a=0|x^*,y=l)-p_{A|XY}(a=0|x^*,y=l') = \epsilon_A$ for some $l' < l$. Since $p_{A|XY}(a=0|x^*,y=l')$ is a term in the vertex of the sub-behavior of $l$ inputs, it must take values in $\{0,\epsilon_A,1-\epsilon_A,1\}$. If this term $p_{A|XY}(a=0|x^*,y=l') = 0$, then the new term $p_{A|XY}(a=0|x^*,y=l) = \epsilon_A$ and the condition is satisfied. If $p_{A|XY}(a=0|x^*,y=l') \neq 0$ giving $p_{A|XY}(a=0|x^*,y=l) \geq 2\epsilon_A$, then it must be the case that $p_{A|XY}(a=0|x^*,y=l'') \neq 0$ for all $l'' < l$. For, if not, if one of these latter terms was $0$, then we would have $p_{A|XY}(a=0|x^*,y=l)-p_{A|XY}(a=0,x^*,y=l'') \geq 2\epsilon_A > \epsilon_A$, violating the PI relaxation condition for inputs $y=l$ and $y=l''$. But then $p_{A|XY}(a=0|x^*,y=l'') \neq 0$ for all $l'' < l$ implies that the sub-behavior for $l$ inputs cannot be a vertex since the PI relaxation submatrix only has rank $l-1$ (at least one of the terms needs to be a $0$ in the vertex). Therefore, we see that in this case also the probabilities only take entries in $\{0,\epsilon_A, 1-\epsilon_A, 1\}$. By mathematical induction, we conclude that the statement holds for arbitrary $|\mathcal{Y}|$.  


\end{proof}

\begin{lem}
\label{lem:vertices-pab-polytope}
Let $\{p^{\lambda_A}_{A|XY}\}_{\lambda_A}$ and $\{p^{\lambda_B}_{B|XY}\}_{\lambda_B}$ denote the sets of vertices of the polytopes $\mathcal{P}_2^{A, \epsilon_A}$ and $\mathcal{P}_2^{B, \epsilon_B}$ respectively. Then $\mathcal{P}_2^{AB, (\epsilon_A, \epsilon_B)}$ is a polytope and its vertices are given by
\begin{equation}
\label{eq:vertices-pab-polytope}
    V_{\mathcal{P}_2^{AB, (\epsilon_A, \epsilon_B)}} := \bigg\{ p^{(\lambda_A,\lambda_B)}_{AB|XY}: p^{(\lambda_A,\lambda_B)}_{AB|XY}(ab|xy) = p^{\lambda_A}_{A|XY}(a|xy) \cdot p^{\lambda_B}_{B|XY}(b|xy) \bigg\}. 
\end{equation}
\end{lem}
\begin{proof}
Firstly we observe that the convex hull of the vertices in $V_{\mathcal{P}_2^{AB, (\epsilon_A, \epsilon_B)}}$ is a subset of $\mathcal{P}_2^{AB, (\epsilon_A, \epsilon_B)}$. This can be seen from the definition of the set $\mathcal{P}_2^{AB, (\epsilon_A, \epsilon_B)}$ in \eqref{eq:P2-PD-polytope}. In particular, by defining $\lambda := (\lambda_A, \lambda_B)$, and $p_{A|XY\Lambda = (\lambda_A,\lambda_B)} := p^{\lambda_A}_{A|XY}$, $p_{B|XY\Lambda = (\lambda_A,\lambda_B)} := p^{\lambda_B}_{B|XY}$, any convex combination with weights $w_{\lambda_A,\lambda_B}$ of these vertices can be seen to belong to $\mathcal{P}_2^{AB, (\epsilon_A, \epsilon_B)}$
by setting $q_{\Lambda}(\lambda_A,\lambda_B) = \sum_{\lambda_A,\lambda_B} \delta(\lambda - (\lambda_A\lambda_B))w_{\lambda_A,\lambda_B}$. 

Similarly, every point $p_{AB|XY} \in \mathcal{P}_2^{AB, (\epsilon_A, \epsilon_B)}$ can be written as a convex combination of the points in the set $V_{\mathcal{P}_2^{AB, (\epsilon_A, \epsilon_B)}}$. To see this, consider any point $p_{AB|XY}$ in $\mathcal{P}_2^{AB, (\epsilon_A, \epsilon_B)}$ written as
\begin{equation}
\label{eq:proof-conv-comb}
    p_{AB|XY}(ab|xy) = \int d\lambda \; q_{\Lambda}(\lambda) \; p_{A|XY\Lambda}(a|xy\lambda) \; p_{B|XY\Lambda}(b|xy\lambda),
\end{equation}
for all $a,b,x,y$ with $q_{\Lambda}(\lambda) \geq 0$ and $\int d\lambda \; q_{\Lambda}(\lambda) = 1$. Now, the $p_{A|XY\Lambda}$ and $p_{B|XY\Lambda}$ can be written as convex combinations of their respective extremal points as
\begin{eqnarray}
\label{eq:conv-comb}
    p_{A|XY\Lambda=\lambda} = \sum_{\lambda_A} r_{\lambda_A}^{\lambda} \; p^{\lambda_A}_{A|XY}, \nonumber \\
 p_{B|XY\Lambda=\lambda} = \sum_{\lambda_B} s_{\lambda_B}^{\lambda} \; p^{\lambda_B}_{B|XY}.
\end{eqnarray}
Here as usual $r^{\lambda}_{\lambda_A} \geq 0$ for all $\lambda, \lambda_A$ and $\sum_{\lambda_A} r^{\lambda}_{\lambda_A} = 1$ for all $\lambda$. Similarly, $s^{\lambda}_{\lambda_B} \geq 0$ for all $\lambda, \lambda_B$ and $\sum_{\lambda_B} s^{\lambda}_{\lambda_B} = 1$ for all $\lambda$. 
So that we obtain by substituting \eqref{eq:conv-comb} into \eqref{eq:proof-conv-comb} that
\begin{equation}
    p_{AB|XY}(ab|xy) = \sum_{\lambda_A,\lambda_B} \alpha_{(\lambda_A,\lambda_B)} \; p^{\lambda_A}_{A|XY}(a|xy) \; p^{\lambda_B}_{B|XY}(b|xy)
\end{equation}
for all $a,b,x,y$ with $\alpha_{(\lambda_A,\lambda_B)}$ defined as
\begin{equation}
    \alpha_{(\lambda_A,\lambda_B)} := \int d\lambda \; q_{\Lambda}(\lambda) \;r^{\lambda}_{\lambda_A} \; s^{\lambda}_{\lambda_B}.
\end{equation}
These coefficients obey the condition that $\alpha_{(\lambda_A, \lambda_B)} \geq 0$ for all $\lambda_A, \lambda_B$ and $\sum_{\lambda_A, \lambda_B} \alpha_{(\lambda_A,\lambda_B)} = 1$. 
This shows the statement
\end{proof}
Now, the polytope $\mathcal{P}_2^{A, \epsilon_A}$ is a Cartesian product of the polytopes $\mathcal{P}_2^{x^*, A, \epsilon_A}$ for all $x^*$ \cite{Grunbaum03}. From Lemma \ref{lem:Alice-polytope-fixedx}, we see that the vertices of $\mathcal{P}_2^{A, \epsilon_A}$ have probabilities $p^{\lambda_A}_{A|XY} \in \{0, \epsilon_A, 1- \epsilon_A, 1\}$. Similarly, the vertices of $\mathcal{P}_2^{B, \epsilon_B}$ have probabilities $p^{\lambda_B}_{B|XY} \in \{0,\epsilon_B, 1-\epsilon_B, 1\}$. Combining with the Lemma \ref{lem:vertices-pab-polytope} proves the theorem. 
\end{proof}

A strictly analogous proof to that of Lemma \ref{lem:vertices-pab-polytope} shows that the $(l,h,\epsilon_A,\epsilon_B)$-MDPDL set is a polytope and its vertices are a subset of the product of the input vertices $p_{XY}$ from \eqref{eq:vertices-inputpolytope} and the vertices $p_{AB|XY}$ from \eqref{eq:vertices-pab-polytope}. This therefore provides a characterisation of the vertices, and by the Minkowski-Weyl theorem, a complete characterisation of the MDPDL set.

\subsection{MDPDL is a Strict Superset of MDL}
By construction, the correlation set defined by $(l,h,\epsilon_A,\epsilon_B)$-MDPDL models in this work strictly contains the $(l,h)$-MDL set studied in~\cite{PRBLG14, PG16}. While it is evident that MDPDL is a superset of MDL due to its more relaxed assumptions, it is important to emphasize that this inclusion is strict. That is, there exists at least one behavior $\{p_{ABXY}\}$ that lies outside the MDL set but inside the MDPDL set. In~\cite{PRBLG14, PG16}, the authors showed that in the $(2,2;2,2)$ Bell scenario, a valid inequality (i.e., a supporting hyperplane) for the $(l,h)$-MDL correlation set is given by:
\begin{equation}\label{eq:md_ineq}
		l \; p_{ABXY}(0000)- h \; \left[p_{ABXY}(0101) + p_{ABXY}(1010) + p_{ABXY}(0011) \right] \leq 0.
	\end{equation}
for any pair of parameters $0 < l \leq h$.
Moreover, this inequality is violated by the quantum behavior $\{p^{H}_{AB|XY}\}$ that achieves the maximal violation of the Hardy paradox~\cite{Hardy93, ZRLH22}. Specifically, this behavior $\{p^{H}_{AB|XY}\}$ satisfies:
\begin{equation}
    p^{H}(01|01)=0,\; p^{H}(10|10)=0,\; p^{H}(00|11)=0,\; \text{and} \; p^{H}(00|00)=\frac{5\sqrt{5}-11}{2} \approx 0.0901.
\end{equation}
It is straightforward to verify that this behavior violates the MDL inequality in Eq.~\eqref{eq:md_ineq} for any input distribution $\{p_{XY}\}$ satisfying $0<l \leq p_{XY}(xy) \leq h,\;\forall x,y$, since the inequality evaluates to:
\begin{equation}
    l \;p^{H}(00|00) p_{XY}(00)\geq l^2 \;p^{H}(00|00)>0.
\end{equation}
Thus, for any input distribution $\{p_{XY}\}$ satisfying $0<l \leq p_{XY}(xy) \leq h$ for all $x,y$, the joint behavior $\{p^{H}_{AB|XY} \cdot p_{XY}\}$ lies outside the $(l,h)$-MDL set. However, the same behavior $\{p^{H}_{AB|XY} \cdot p_{XY}\}$ (for any input distribution ${p_{XY}}$ satisfying $0< l \leq p_{XY}(xy) \leq h$ for all $x,y$) lies within the $(l,h,\epsilon,\epsilon)$-MDPDL set for sufficiently small $\epsilon \ll 1$. This is because $\{p^{H}_{AB|XY}\}$  can be decomposed as a convex combination of the extreme points defined in \eqref{eq:v_pd} of the polytope $\mathcal{P}_2^{AB,(l,h,\epsilon,\epsilon)}$.

We list the behavior $\{p^{H}_{AB|XY}\}$ in the first row of the following table. Its convex decomposition into extreme points of $\mathcal{P}_2^{AB,(l,h,\epsilon,\epsilon)}$, along with their corresponding weights $\{w_i\}$, is provided in the subsequent rows:
\begin{table}[H]
        \centering
        \begin{tabular}{c| c c c c  c c c c}
        \hline
        $\{p^{H}_{AB|XY}\}$ & $p^{H}(00|00)$ & $p^{H}(01|00)$ & $p^{H}(10|00)$ & $p^{H}(11|00)$ & $p^{H}(00|01)$ & $p^{H}(01|01)$ & $p^{H}(10|01)$ & $p^{H}(11|01)$ \\
        \hline 
          & $\frac{5\sqrt{5}-11}{2}$ & $\frac{7-3\sqrt{5}}{2}$ & $\frac{7-3\sqrt{5}}{2}$ & $\frac{\sqrt{5}-1}{2}$ & $\sqrt{5}-2$
           & 0 & $\frac{7-3\sqrt{5}}{2}$ & $\frac{\sqrt{5}-1}{2}$   \\
        \hline
        & $p^{H}(00|10)$ & $p^{H}(01|10)$ & $p^{H}(10|10)$ & $p^{H}(11|10)$ & $p^{H}(00|11)$ & $p^{H}(01|11)$ & $p^{H}(10|11)$ & $p^{H}(11|11)$ \\
        \hline
          & $\sqrt{5}-2$ & $\frac{7-3\sqrt{5}}{2}$  & 0 & $\frac{\sqrt{5}-1}{2}$ & 0
            & $\frac{3-\sqrt{5}}{2}$ & $\frac{3-\sqrt{5}}{2}$ & $\sqrt{5}-2$ \\
        \hline
        \multicolumn{9}{c}{Convex decomposition of the behavior $\{p^{H}_{AB|XY}\}$ into  extreme points~\eqref{eq:v_pd} with $\epsilon=\frac{1}{4}$} \\
        \hline
        weight & $p^{ext}_{A|XY}(0|00)$ & 
        $p^{ext}_{B|XY}(0|00)$ & $p^{ext}_{A|XY}(0|01)$ & 
        $p^{ext}_{B|XY}(0|01)$ &
        $p^{ext}_{A|XY}(0|10)$ & 
        $p^{ext}_{B|XY}(0|10)$ & $p^{ext}_{A|XY}(0|11)$ &  $p^{ext}_{B|XY}(0|11)$ \\
        \hline 
        $w_1=\frac{108\sqrt{5}-241}{5}$ & 0 & 0 & 0 & 0 &
          0  & 0 & $\epsilon$ & 0 \\
        \hline
        $w_2=\frac{108\sqrt{5}-241}{5}$ & 0 & 0 & 0 & 0 &
           $\epsilon$ & 0 & 0 & 0 \\
        \hline
        $w_3=\frac{311-138\sqrt{5}}{15}$ & 0 & 0 & 0 & 0 &
          1  & $\epsilon$ & 1 & 0 \\
        \hline
        $w_4=\frac{14-6\sqrt{5}}{3}$ & 0 & 0 & $\epsilon$ & 1 &
          0  & 0 & 0 & 1 \\
        \hline
        $w_5=\frac{1647-736\sqrt{5}}{20}$ & $\epsilon$ & $\epsilon$ & 0 & 0 &
           0 & 0 & 0 & 0 \\
        \hline
        $w_6=\frac{1169-522\sqrt{5}}{20}$ & $\epsilon$ & $1-\epsilon$ & 0 & 0 &
          1 & 1 & 1 & 0 \\
        \hline
        $w_7=\frac{882\sqrt{5}-1969}{30}$ & $\epsilon$  & 1 & 0 & 0 &
          1  & 1 & 1 & 0 \\
        \hline
        $w_8=\frac{141-63\sqrt{5}}{2}$ & $1-\epsilon$ & $\epsilon$ & 1 & 1 & 
           0 & 0 & 0 & 1 \\
        \hline
        $w_9=\frac{99\sqrt{5}-221}{3}$ & 1 & $\epsilon$ & 1 & 1 &
          0  & 0 & 0 &  1\\
        \hline
        \end{tabular}
            \caption{Convex decomposition of the behavior ${p^{H}_{AB|XY}}$ into extreme points of the set $\mathcal{P}_2^{AB,(l,h,\epsilon,\epsilon)}$ with $\epsilon=\frac{1}{4}$.}
            \label{tab:placeholder}
        \end{table}

\section{Facet inequality of the $\epsilon$-Parameter dependent polytope $\mathcal{P}^{AB,(\epsilon, \epsilon)}_2$ in the $(2,2;2,2)$ Bell scenario and its quantum violation}\label{app:facet_pd}
\subsection{Facet inequality of the $\epsilon$-Parameter dependent polytope $\mathcal{P}^{AB,(\epsilon, \epsilon)}_2$ in the $(2,2;2,2)$ Bell scenario}
    The $\epsilon$-Parameter dependent polytope for any $\epsilon \in [0,1)$ is defined as
	\begin{equation}\label{eq:def_pd}
		\begin{split}
			\mathcal{P}^{AB,(\epsilon, \epsilon)}_2 := \{& p_{AB|XY} : p_{AB|XY}(ab|xy)  = \int d \lambda \; q_{\Lambda}(\lambda)  \; p_{A|XY\Lambda}(a|xy\lambda) \; p_{B|XY\Lambda}(b|xy\lambda),\\
			& \frac{1}{2} \sum_a \bigg| p_{A|XY\Lambda}(a|xy\lambda) - p_{A|XY\Lambda}(a|xy'\lambda) \bigg| \leq \epsilon\; \; \forall x,\lambda, y,y',\\
			& \frac{1}{2} \sum_b \bigg| p_{B|XY\Lambda}(b|xy\lambda) - p_{B|XY\Lambda}(b|x'y\lambda) \bigg| \leq \epsilon \;\; \forall y,\lambda, x,x'. \\
			& q_{\Lambda}(\lambda)\geq 0,\forall \lambda, \int d \lambda q_{\Lambda}(\lambda)=1\}.
		\end{split}
	\end{equation}
Combining this with the result proven in App.~\ref{app:MDPDL-polytope}, we obtain the vertices of the $\mathcal{P}^{AB,(\epsilon, \epsilon)}_2$ polytope:
	\begin{equation}\label{eq:v_pd}
		\begin{split}
			\mathcal{V}_{\mathcal{P}^{AB,(\epsilon, \epsilon)}_2}  =\big\{& p_{AB|XY} :  p_{AB|XY}(ab|xy)  = p_{A|XY}(a|xy)p_{B|XY}(b|xy) \\
			& \text{with } p_{A|XY}(a|xy)\in\{0,1,\epsilon,1-\epsilon\}, \; \frac{1}{2} \sum_a \big| p_{A|XY}(a|xy) - p_{A|XY}(a|xy') \big|\leq \epsilon\;\forall x, y,y', \\
			& \text{and } p_{B|XY}(b|xy)\in\{0,1,\epsilon,1-\epsilon\}, \;\frac{1}{2} \sum_b \big| p_{B|XY}(b|xy) - p_{B|XY}(b|x'y) \big|\leq \epsilon\;\forall y, x,x'.\big\}\\
		\end{split}
	\end{equation}
In the $(2,2;2,2)$ Bell scenario, the variables $a, b, x, y$ take values in $\{0,1\}$. Hence, the $\epsilon$-PD constraints in Eq.~\eqref{eq:v_pd} reduce to:
	\begin{equation}
		\begin{split}
			& \frac{1}{2}\left| p_{A|XY}(a=0|x,y=0) - p_{A|XY}(a=0|x,y'=1)\right| +\frac{1}{2}\left| p_{A|XY}(a=1|x,y=0) - p_{A|XY}(a=1|x,y'=1)\right| \\
			= & \left| p_{A|XY}(a=0|x,y=0) - p_{A|XY}(a=0|x,y'=1)\right| \leq \epsilon\; \; \;\;\forall x\in\{0,1\}
		\end{split}
	\end{equation}
	and similarly,
	\begin{equation}
		\left| p_{B|XY}(b=0|x=0,y) - p_{B|XY}(b=0|x'=1,y)\right| \leq \epsilon\; \; \;\;\forall y\in\{0,1\}
	\end{equation}
    In the above derivation, we used the normalization conditions $\sum_a p_{A|XY}(a|x,y) = 1$ and $\sum_b p_{B|XY}(b|x,y) = 1$ for all $x,y \in \{0,1\}$. Therefore, in the $(2,2;2,2)$ Bell scenario, the $\epsilon$-PD polytope $\mathcal{P}^{AB,(\epsilon, \epsilon)}_2(2,2;2,2)$ contains $1296$ vertices, each of the form:
	\begin{equation}\label{eq:v_pd}
		\begin{split}
			&\mathcal{V}_{\mathcal{P}^{AB,(\epsilon, \epsilon)}_2(2,2;2,2)} \\
            &=\big\{ p_{AB|XY} :  p_{AB|XY}(ab|xy)  = p_{A|XY}(a|xy)p_{B|XY}(b|xy) \\
			&\qquad \text{with } \left(p_{A|XY}(0|x,0),\;p_{A|XY}(0|x,1)\right)\in\big \{(0,0),(1,1),(0,\epsilon),(1,1-\epsilon),(\epsilon,0),(1-\epsilon,1)\big\}, \; \forall x\in\{0,1\}\\
			& \qquad\text{and } \left(p_{B|XY}(0|0,y), \; p_{B|XY}(0|1,y)\right)\in\big \{(0,0),(1,1),(0,\epsilon),(1,1-\epsilon),(\epsilon,0),(1-\epsilon,1)\big\}, \; \forall y\in\{0,1\} \big\}\\
		\end{split}
	\end{equation}
	Note that the extreme points do not take the values
    \begin{equation}    \left(p_{A|XY}(0|x,0), \;p_{A|XY}(0|x,1)\right)=(\epsilon,\epsilon) \text{ or } \left(p_{A|XY}(0|x,0), \;p_{A|XY}(0|x,1)\right)=(1-\epsilon,1-\epsilon).
    \end{equation}
    This is because any vertex with $\left(p_{A|XY}(0|x,0), \; p_{A|XY}(0|x,1)\right) = (\epsilon, \epsilon)$ can be decomposed as a convex combination of the deterministic vertices $\left(p_{A|XY}(0|x,0), \; p_{A|XY}(0|x,1)\right) = (1,1)$ with weight $\epsilon$ and $(0,0)$ with weight $1-\epsilon$.
\begin{theorem}
        For any $\epsilon\in[0,1)$, the following inequality defines a facet of the $\epsilon$-PD polytope $\mathcal{P}^{AB,(\epsilon, \epsilon)}_2(2,2;2,2)$:
	\begin{equation}\label{eq:pd_ineq}
		(1-\epsilon)\;p_{AB|XY}(00|00)+\epsilon(1-\epsilon) \;p_{AB|XY}(11|00)- p_{AB|XY}(01|01) - p_{AB|XY}(10|10) - p_{AB|XY}(00|11)\leq \epsilon(1-\epsilon).
	\end{equation}
\end{theorem}
    \begin{proof}
    First, the upper bound in inequality~\eqref{eq:pd_ineq} is verified by explicitly evaluating the left-hand side over all vertices listed in Eq.~\eqref{eq:v_pd}. Next, we aim to prove that Eq.~\eqref{eq:pd_ineq} defines a facet of the polytope $\mathcal{P}^{AB,(\epsilon, \epsilon)}_2(2,2;2,2)$. In other words, we will show that all vertices from Eq.~\eqref{eq:v_pd} that saturate the inequality compose a hyperplane of dimension $\text{dim}\left[\mathcal{P}^{AB,(\epsilon, \epsilon)}_2(2,2;2,2) \right] - 1$.

    For any $\epsilon \in (0,1)$, the dimension of the polytope is $\text{dim}\left[\mathcal{P}^{AB,(\epsilon, \epsilon)}_2(2,2;2,2)\right] = 12$, due to the four normalization constraints. When $\epsilon = 0$, the $\epsilon$-PD conditions in Eq.~\eqref{eq:def_pd} and Eq.~\eqref{eq:v_pd} reduce to the standard no-signaling constraints, and in that case, $\text{dim}\left[\mathcal{P}^{AB,(\epsilon=0, \epsilon=0)}_2(2,2;2,2) \right] = 8$. Since the $\epsilon = 0$ case is well studied in the literature, we focus on the nontrivial regime $\epsilon \in (0,1)$ in the following.

    There are a total of $56$ vertices that saturate the upper bound of Eq.~\eqref{eq:pd_ineq}. These $56$ vertices can be classified into five distinct types. (see \cite{github} for an explicit list of the vertices.)
    
    \begin{enumerate}
        \item [1.] 
        The first type of vertices satisfy $p_{AB|XY}(00|00)=0,\; p_{AB|XY}(11|00)=1$ and $ p_{AB|XY}(01|01)=0,\; p_{AB|XY}(10|10)=0, \; p_{AB|XY}(00|11)=0.$ Such vertices must have $p_{A|XY}(0|00)=0$ and $p_{B|XY}(0|00)=0$. There are $28$ vertices of this type that saturate the upper bound of Eq.~\eqref{eq:pd_ineq}. 

        \item [2.] 
        The second type of vertices satisfy $p_{AB|XY}(00|00)=\epsilon,\; p_{AB|XY}(11|00)=0$ and $ p_{AB|XY}(01|01)=0,\; p_{AB|XY}(10|10)=0, \; p_{AB|XY}(00|11)=0.$ These vertices must satisfy either $p_{A|XY}(0|00)=\epsilon,\; p_{B|XY}(0|00)=1$ or $p_{A|X,Y}(0|00)=1,\; p_{B|X,Y}(0|00)=\epsilon$. There are $16$ such vertices saturating the upper bound of Eq.~\eqref{eq:pd_ineq}. 

        \item [3.] 
        The third type of vertices satisfy $p_{AB|XY}(00|00)=1,\; p_{AB|XY}(11|00)=0$ and $ p_{AB|XY}(01|01)=(1-\epsilon)^2,\; p_{AB|XY}(10|10)=0, \; p_{AB|XY}(00|11)=0.$  These vertices must have $p_{A|XY}(0|00)=1$, $_{B|XY}(0|00)=1$, and simultaneously $p_{A|XY}(0|01)=1-\epsilon$, $p_{B|XY}(0|01)=\epsilon$. There are $4$ vertices of this type that saturate the upper bound of Eq.~\eqref{eq:pd_ineq}. 
        

        \item [4.] 
        The fourth type of vertices satisfy $p_{AB|XY}(00|00)=1,\; p_{AB|XY}(11|00)=0$ and $ p_{AB|XY}(01|01)=0,\; p_{AB|XY}(10|10)=(1-\epsilon)^2, \; p_{AB|XY}(00|11)=0.$ These vertices must have $p_{A|XY}(0|00)=1$, $p_{B|XY}(0|00)=1$, and simultaneously $p_{A|XY}(0|10)=\epsilon$, $p_{B|XY}(0|10)=1-\epsilon$. There are $4$ such vertices that saturate the upper bound of Eq.~\eqref{eq:pd_ineq}. 
        

        \item [5.] 
        The fifth type of vertices satisfy $p_{AB|XY}(00|00)=1,\; p_{AB|XY}(11|00)=0$ and $ p_{AB|XY}(01|01)=0,\; p_{AB|XY}(10|10)=0, \; p_{AB|XY}(00|11)=(1-\epsilon)^2.$ These vertices must have $p_{A|XY}(0|00)=1$, $p_{B|XY}(0|00)=1$, and simultaneously $p_{A|XY}(0|11)=1-\epsilon$, $p_{B|XY}(0|11)=1-\epsilon$. There are $4$ such vertices that saturate the upper bound of Eq.~\eqref{eq:pd_ineq}. 
        
    \end{enumerate}
For any $\epsilon \in (0,1)$, these $56$ vertices that saturate the upper bound of Eq.~\eqref{eq:pd_ineq} lie on a hyperplane of dimension $11$, as can be verified by noting that $12$ vertices ($V_1, V_2, V_7, V_{10}, V_{24}, V_{26}, V_{29}, V_{37}, V_{44}, V_{45}, V_{51}, V_{53}$ in \cite{github}) are affinely independent.   

To be clear, the entries for these $12$ vertices are listed in the table below (subscripts of $p_{AB|XY}$ suppressed in the table for compactness). The entries $p_{AB|XY}(00|xy)$ for all $x,y\in\{0,1\}$ are omitted by normalization. Each remaining entry is computed via $p_{AB|XY}(01|xy)=p_{A|XY}(0|xy)\big(1-p_{B|XY}(0|xy)\big)$, $p_{AB|XY}(10|xy)=\big(1-p_{A|XY}(0|xy)\big)p_{B|XY}(0|xy)$, and $p_{AB|XY}(11|xy)=\big(1-p_{A|XY}(0|xy)\big)\big(1-p_{B|XY}(0|xy)\big)$. Viewing the table as a $12\times12$ matrix $M_{V}$, these $12$ vertices are affinely independent because the matrix $M_{V}$ has $\text{rank}(M_{V})=12$, which is established by its determinant being strictly positive for any $\epsilon\in(0,1)$:
    \begin{equation}
    \text{det}(M_{V})=4(2-\epsilon)(1-\epsilon)^6\epsilon^5>0.
    \end{equation}
    Therefore, the inequality in Eq.~\eqref{eq:pd_ineq} defines a facet of the polytope $\mathcal{P}^{AB,(\epsilon, \epsilon)}_2(2,2;2,2)$.

\begin{table}[H]
\centering
\setlength{\tabcolsep}{4pt}
\scalebox{0.9}{\begin{tabular}{c|ccc|ccc|ccc|ccc}
\hline
Index &
$p(01|00)$ & $p(10|00)$ & $p(11|00)$ &
$p(01|01)$ & $p(10|01)$ & $p(11|01)$ &
$p(01|10)$ & $p(10|10)$ & $p(11|10)$ &
$p(01|11)$ & $p(10|11)$ & $p(11|11)$ \\
\hline
$V_{1}$  & $0$ & $0$ & $1$ & $0$ & $0$ & $1$ & $0$ & $0$ & $1$ & $0$ & $0$ & $1$ \\
$V_{2}$  & $0$ & $0$ & $1$ & $0$ & $1$ & $0$ & $0$ & $0$ & $1$ & $0$ & $1$ & $0$ \\
$V_{7}$  & $0$ & $0$ & $1$ & $0$ & $0$ & $1$ & $1$ & $0$ & $0$ & $1$ & $0$ & $0$ \\
$V_{10}$ & $0$ & $0$ & $1$ & $0$ & $\epsilon$ & $1-\epsilon$ & $0$ & $0$ & $1$ & $\epsilon$ & $0$ & $1-\epsilon$ \\
$V_{24}$ & $0$ & $0$ & $1$ & $0$ & $1-\epsilon$ & $0$ & $\epsilon$ & $0$ & $1-\epsilon$ & $0$ & $1-\epsilon$ & $\epsilon$ \\
$V_{26}$ & $0$ & $0$ & $1$ & $0$ & $\epsilon$ & $1-\epsilon$ & $1-\epsilon$ & $0$ & $0$ & $1$ & $0$ & $0$ \\
$V_{29}$ & $0$ & $1-\epsilon$ & $0$ & $0$ & $0$ & $1$ & $0$ & $0$ & $0$ & $1$ & $0$ & $0$ \\
$V_{37}$ & $1-\epsilon$ & $0$ & $0$ & $0$ & $0$ & $0$ & $0$ & $0$ & $1$ & $0$ & $1$ & $0$ \\
$V_{44}$ & $1-\epsilon$ & $0$ & $0$ & $0$ & $\epsilon$ & $0$ & $\epsilon$ & $0$ & $1-\epsilon$ & $0$ & $1-\epsilon$ & $\epsilon$ \\
$V_{45}$ & $0$ & $0$ & $0$ & $(1-\epsilon)^2$ & $\epsilon^2$ & $\epsilon-\epsilon^2$ & $0$ & $0$ & $0$ & $1$ & $0$ & $0$ \\
$V_{51}$ & $0$ & $0$ & $0$ & $0$ & $\epsilon$ & $0$ & $\epsilon^2$ & $(1-\epsilon)^2$ & $\epsilon-\epsilon^2$ & $0$ & $0$ & $1$ \\
$V_{53}$ & $0$ & $0$ & $0$ & $0$ & $0$ & $0$ & $0$ & $0$ & $0$ & $\epsilon-\epsilon^2$ & $\epsilon-\epsilon^2$ & $\epsilon^2$ \\
\hline
\end{tabular}}
\caption{The 12 affinely independent vertices that saturate the upper bound of Eq.~\eqref{eq:pd_ineq}.}
\label{}
\end{table}
\end{proof}
\subsection{Quantum violation of the Facet Inequality Eq.~\eqref{eq:pd_ineq}}

    In this section, we prove that for any $\epsilon \in [0,1)$, the facet inequality Eq.~\eqref{eq:pd_ineq} can be violated by a quantum correlation arising from a two-qubit system. This system is closely related to the Tilted Hardy paradox introduced by us in~\cite{ZRLH22}.

    \begin{lemma}
        For any $\epsilon \in [0,1)$ the facet inequality Eq.~\eqref{eq:pd_ineq} can be violated by the two-qubit partially entangled state 
        \begin{equation}
        |\psi_{\theta}\rangle=\cos(\frac{\theta}{2})|00\rangle-\sin(\frac{\theta}{2})|11\rangle
        \end{equation}
        with the parameter $\theta=\arcsin{(3-\sqrt{4\epsilon+5})}$, along with the observables:
        \begin{equation}\label{eq:q_strategy}
        \begin{split}
        A_0=B_0&=\frac{-(2+\sin\theta)(1-\sin\theta)^{\frac{1}{2}}}{(2-\sin\theta)(1+\sin\theta)^{\frac{1}{2}}} \sigma_z+\frac{-\sqrt2\sin\theta(\sin\theta)^{\frac{1}{2}}}{(2-\sin\theta)(1+\sin\theta)^{\frac{1}{2}}} \sigma_x \nonumber, \\
        A_1=B_1&=\frac{-(1-\sin\theta)^{\frac{1}{2}}}{(1+\sin\theta)^{\frac{1}{2}}}\sigma_z+\frac{\sqrt2(\sin\theta)^{\frac{1}{2}}}{(1+\sin\theta)^{\frac{1}{2}}}\sigma_x ,
    \end{split}
    \end{equation}
    where $\sigma_x$ and $\sigma_z$ are the Pauli operators.
    \end{lemma}
    
    \begin{proof}
    This quantum strategy yields the following correlations:
\begin{equation}
\begin{split}
		& p_{AB|XY}(01|01)=0,\; p_{AB|XY}(10|10)=0,\; p_{AB|XY}(00|11)=0,\;\\
        & p_{AB|XY}(00|00)+\epsilon \;p_{AB|XY}(11|00)= \frac{(4\epsilon+5)\sqrt{4\epsilon+5}-(12\epsilon+11)}{2(1+\epsilon)}> \epsilon,\; \forall \epsilon\in[0,1). 
\end{split}
\end{equation}
Therefore, this quantum strategy leads to a strict violation of the facet inequality Eq.~\eqref{eq:pd_ineq}, as
\begin{equation}
    \left(1-\epsilon\right)\left[\frac{(4\epsilon+5)\sqrt{4\epsilon+5}-(12\epsilon+11)}{2(1+\epsilon)}-\epsilon\right]>0,\; \forall \epsilon\in[0,1).
\end{equation}
\end{proof}
We thus see that for any $\epsilon \in [0,1)$, there exists a quantum correlation that lies outside the $\epsilon$-parameter dependent set $\mathcal{P}^{AB,(\epsilon, \epsilon)}_2(2,2;2,2)$.

\section{Supporting hyperplane of the MDPDL set in the $(2,2;2,2)$ Bell scenario and its quantum violation}
\label{app:MDPDL-ineq}
In this section, we derive a supporting hyperplane of the MDPDL set in the $(2,2;2,2)$ Bell scenario-i.e., an MDPDL inequality-and demonstrate its quantum violation based on the results discussed in the previous sections.

\begin{theorem}
 For any MD parameters $0 < l \leq h$ and any PD parameter $0 \leq \epsilon < 1$, the following inequality holds for any behavior in the $(l,h,\epsilon,\epsilon)$-MDPDL set:
    \begin{equation}\label{eq:mdpd_ineq}
		l(1-\epsilon) \; \left[p_{ABXY}(0000)+\epsilon p_{ABXY}(1100) -\epsilon p_{XY}(00) \right]- h \; \left[p_{ABXY}(0101) + p_{ABXY}(1010) + p_{ABXY}(0011) \right] \leq 0.
	\end{equation}
This MDPDL inequality can be violated by the quantum strategy defined in Eq.~\eqref{eq:q_strategy} along with any input distribution $\{p_{XY}\}$ satisfying $0<l\leq p_{XY}(xy)\leq h,\forall x,y.$
\end{theorem}
\begin{proof}

We first explain why the upper bound of this inequality Eq.~\eqref{eq:mdpd_ineq} over the MDPDL correlation set is 0.
    
As shown in App.~\ref{app:MDPDL-polytope}, 
any extremal point $p^{ext}_{ABXY}$ of the MDPDL set can be written as $p^{ext}_{ABXY}(abxy) = p^{ext}_{AB|XY}(ab|xy) \cdot p^{ext}_{XY}(xy)$, where $p^{ext}_{AB|XY}$ satisfies the $(\epsilon,\epsilon)$-PD constraints in Eq.~\eqref{eq:v_pd} and $p^{ext}_{XY}$ satisfies the $(l,h)$-MD constraints: $l \leq p^{ext}_{XY}(xy) \leq h$ for any $x,y$. For any such vertex $p^{ext}_{ABXY}$, the left-hand side of Eq.~\eqref{eq:mdpd_ineq} becomes:
\begin{equation}
    \begin{split}
		&l(1-\epsilon)\left[p^{ext}_{AB|XY}(00|00)+\epsilon \;p^{ext}_{AB|XY}(11|00) -\epsilon \right] p^{ext}_{XY}(00)\\
        & - h\left[p^{ext}_{AB|XY}(01|01) \; p^{ext}_{XY}(01) + p^{ext}_{AB|XY}(10|10)\; p^{ext}_{XY}(10) + p^{ext}_{AB|XY}(00|11) \; p^{ext}_{XY}(11)\right] \\
        & \leq lh(1-\epsilon)\left[p^{ext}_{AB|XY}(00|00)+\epsilon \; p^{ext}_{AB|XY}(11|00) -\epsilon \right]
         - lh\left[p^{ext}_{AB|XY}(01|01) + p^{ext}_{AB|XY}(10|10) + p^{ext}_{AB|XY}(00|11)\right]\\
         & \leq 0
	\end{split}
\end{equation}
where the first inequality follows from the MD-$(l,h)$ constraints and the second from the facet inequality Eq.~\eqref{eq:pd_ineq} derived in App.~\ref{app:facet_pd}.
Now, consider the two-qubit quantum strategy defined in Eq.~\eqref{eq:q_strategy} for any $\epsilon \in [0,1)$. This strategy produces a strict violation of the MDPDL inequality Eq.~\eqref{eq:mdpd_ineq} for any $0 < l \leq h$. Evaluating the left-hand side of Eq.~\eqref{eq:mdpd_ineq} under this strategy yields:
\begin{equation}
\begin{split}
    & l \; p^{ext}_{XY}(00)\left(1-\epsilon\right)\left[\frac{(4\epsilon+5)\sqrt{4\epsilon+5}-(12\epsilon+11)}{2(1+\epsilon)}-\epsilon\right]\\
    & \geq l^2\left(1-\epsilon\right)\left[\frac{(4\epsilon+5)\sqrt{4\epsilon+5}-(12\epsilon+11)}{2(1+\epsilon)}-\epsilon\right]>0,
\end{split}
\end{equation}
$\forall 0<l\leq h,\; 0\leq \epsilon<1$.
This proves that for any MD parameters $0 < l \leq h$ and any PD parameter $0 \leq \epsilon < 1$, there exists a quantum correlation that lies outside the ($l,h,\epsilon,\epsilon)$-MDPDL correlation set.
\end{proof}
\section{Guessing probability in the CHSH Bell test under leakage of input information}
\label{app:Pg-crosstalk}
    In this section, we derive an analytical expression for the guessing probability of a quantum adversary, Eve, who attempts to guess the output of a fixed measurement setting of one party, Alice, in the CHSH Bell test. In this setting, Alice's input information is partially leaked to the other party, Bob, such that the mutual information between Alice's input $X$ and Bob's output $B$ is $I(X:B) = \kappa$, with $\kappa \in [0,1)$. We begin by explaining how input leakage affects the CHSH Bell test, and then derive Eve's guessing probability under this modified setting.

    \subsection{CHSH Bell test under partial leakage of Alice's input information}
    The CHSH Bell test is defined in the $(2,2;2,2)$ Bell scenario, where Alice and Bob each perform one of two binary-outcome measurements, denoted $A_x$ for $x \in \{0,1\}$ and $B_y$ for $y \in \{0,1\}$. The CHSH Bell expression is given by:
    \begin{equation}\label{eq:chsh}
    I_{\text{CHSH}} := \langle A_0 B_0 \rangle + \langle A_0 B_1 \rangle + \langle A_1 B_0 \rangle - \langle A_1 B_1 \rangle,
    \end{equation}
    where the correlators are defined as $\langle A_x B_y \rangle := \sum_{a,b \in \{0,1\}} (-1)^{a+b} \cdot p_{AB|XY}(a,b|A_x,B_y)$. It is important to note that in this $(2,2;2,2)$ Bell scenario, the classical bound under input leakage from Alice to Bob is equivalent to the classical bound in the PD-$(\epsilon_A, \epsilon_B)$ model with $\epsilon_A = 0$ and $\epsilon_B = \kappa$. In App.~\ref{app:facet_pd}, we listed all extremal classical strategies in this setting, we use these to verify that the classical bound for this inequality is $S_{c,\kappa} := 2 + 2\kappa$.
Nevertheless, for quantum correlations, when Alice's input $x \in \{0,1\}$ is leaked to Bob, Bob may adapt his measurement depending on this information. That is, Bob's observables become $B_{y,x}$, and the CHSH expression in Eq.~\eqref{eq:chsh} is modified to:
    \begin{equation}
        S_{\text{CHSH}} = \langle A_0 B_{0,0} \rangle + \langle A_0 B_{1,0} \rangle + \langle A_1 B_{0,1} \rangle - \langle A_1 B_{1,1} \rangle.
    \end{equation}
    Following the approach of~\cite{SPM13}, the condition that Alice's input is only partially leaked to Bob (i.e., $I(X:B) = \kappa$ with $\kappa \in [0,1)$) can be translated into the following constraint:
    \begin{equation}
        -\kappa\mathds{1} \leq \Pi_{y,x=0}^{b}-\Pi_{y,x=1}^{b}\leq \kappa\mathds{1},\;\forall b,y\in\{0,1\}
    \end{equation}
    where $\Pi_{y,x}^{b}$ are POVM elements corresponding to $B_{y,x}$, satisfying $\sum_{b \in {0,1}} \Pi_{y,x}^{b} = \mathds{1}$, and the observables are defined as $B_{y,x} = \sum_{b \in {0,1}} (-1)^b \Pi_{y,x}^{b}$.
    This leads to the following constraint on the operator norm of the difference between observables:
    \begin{equation}
        \|B_{y,x=0}-B_{y,x=1}\| \leq 2\kappa,\; \forall y\in\{0,1\}.
    \end{equation}
Putting everything together, the CHSH Bell inequality under input leakage becomes:
    \begin{equation}\label{eq:s_chsh}
        \begin{split}
            S_{\text{CHSH}} & = \langle A_0 B_{0,0} \rangle + \langle A_0 B_{1,0} \rangle + \langle A_1 B_{0,1} \rangle - \langle A_1 B_{1,1} \rangle \leq 2+2\kappa,\\
            &\|B_{y,x=0}-B_{y,x=1}\| \leq 2\kappa,\; \forall y\in\{0,1\}.
        \end{split}
    \end{equation}
    
     \begin{theorem}
     The quantum value of the CHSH Bell inequality under $\kappa$ bits of input leakage, as defined in Eq.~\eqref{eq:s_chsh}, is given by
        \begin{equation}
        S_{q,\kappa}=\begin{cases}
        2\sqrt{2}(\kappa+\sqrt{1-\kappa^2}),\; & \text{ when } \kappa\in[0,\frac{1}{\sqrt{2}}];\\
        4, & \text{ when } \kappa\in[\frac{1}{\sqrt{2}},1).
        \end{cases}
        \end{equation}
        This value is achieved by the maximally entangled two-qubit state $|\psi\rangle = \frac{1}{\sqrt{2}}(|00\rangle + |11\rangle)$ together with the following observables:
        \begin{align}
        &\text{ when } \kappa\in[0,\frac{1}{\sqrt{2}}],\quad 
        \begin{cases}
        &A_0=\sigma_z,\; A_1=\sigma_x;\\
        &B_{0,0}=\frac{-\kappa+\sqrt{1-\kappa^2}}{\sqrt{2}}\sigma_x+\frac{\kappa+\sqrt{1-\kappa^2}}{\sqrt{2}}\sigma_z,\\
        &B_{0,1}=\frac{\kappa+\sqrt{1-\kappa^2}}{\sqrt{2}}\sigma_x+\frac{-\kappa+\sqrt{1-\kappa^2}}{\sqrt{2}}\sigma_z,\\
        &B_{1,0}=\frac{\kappa-\sqrt{1-\kappa^2}}{\sqrt{2}}\sigma_x+\frac{\kappa+\sqrt{1-\kappa^2}}{\sqrt{2}}\sigma_z,\\
        &B_{1,1}=\frac{-\kappa-\sqrt{1-\kappa^2}}{\sqrt{2}}\sigma_x+\frac{-\kappa+\sqrt{1-\kappa^2}}{\sqrt{2}}\sigma_z.
        \end{cases}\\
        & \text{ when } \kappa\in[\frac{1}{\sqrt{2}},1),\quad 
        \begin{cases}
        &A_0=\sigma_z,\; A_1=\sigma_x;\\
        &B_{0,0}=\sigma_z,\;B_{0,1}=\sigma_x,\\
        &B_{1,0}=\sigma_z,\;B_{1,1}=-\sigma_x.
        \end{cases}
        \end{align}
    \end{theorem}
    The proof of this theorem is provided in \ref{pf:Sq_pg}.
    \subsection{Eve's guessing probability in the CHSH Bell test under partial leakage of Alice's input information}
    \begin{theorem}
    Let $S \in [S_{c,\kappa}, S_{q,\kappa}]$ be the observed CHSH value under $\kappa$ bits of input leakage in the Bell test defined in Eq.~\eqref{eq:s_chsh}. Then the quantum adversary's guessing probability for the outcome of Alice's measurement setting $x^*=0$ is given by:
        \begin{equation}
        {P}_{g}^{(q)}(S,\kappa,x^*=0)= \begin{cases}
        \overline{P}_{g}^{(q)}(S,\kappa,x^*=0), & \text{ when } S\in[S_{\kappa}^*,S_{q,\kappa}];\\
        \frac{\overline{P}_{g}^{(q)}(S=S_{\kappa}^*,\kappa,x^*=0)-1}{S_{\kappa}^*-S_{c,\kappa}} S+\frac{S_{\kappa}^*-S_{c,\kappa}\overline{P}_{g}^{(q)}(S=S_{\kappa}^*,\kappa,x^*=0)}{S_{\kappa}^*-S_{c,\kappa}}, & \text{ when } S\in[S_{c,\kappa},S_{\kappa}^*].\\
        \end{cases}
        \end{equation}
        Here, the function $\overline{P}{g}^{(q)}(S,\kappa,x^*=0)$ is given by:
        \begin{equation}
            \overline{P}{g}^{(q)}(S,\kappa,x^*=0)=\begin{cases}
            \frac{1}{2}+\frac{1}{4}\sqrt{4-\big( \sqrt{S^2-(2-4\kappa^2)^2} -4\kappa\sqrt{1-\kappa^2}\big)^2}, & \text{ when } \kappa\in[0,\frac{1}{\sqrt{2}}];\\
            \frac{1}{2}+\frac{1}{4}\sqrt{S(4-S)}, & \text{ when } \kappa\in[\frac{1}{\sqrt{2}},1).
        \end{cases}
        \end{equation}
        The transition point $S_{\kappa}^*$ is the unique real solution in the interval $[S_{c,\kappa},S_{q,\kappa}]$ of the equation:
        \begin{equation}
            \frac{\partial \overline{P}_{g}^{(q)}(S,\kappa,x^*=0)}{\partial S}\left(S-S_{c, \kappa}\right)=\overline{P}_{g}^{(q)}(S,\kappa,x^*=0)-1.
        \end{equation}
    \end{theorem}
   \begin{proof}\label{pf:Sq_pg}
    To derive Eve's guessing probability for the outcome of a fixed measurement setting of Alice, we assume without loss of generality that the adversary aims to guess the outcome corresponding to the measurement setting $A_0$. It is known that when the adversary's goal is to guess the outcome of a fixed measurement, the optimal guessing probability for a quantum adversary is identical to that for a classical one (see App.~A of~\cite{RLWP25} for a detailed explanation). Therefore, the guessing probability in the CHSH Bell test with input leakage can be derived from the monogamy relation between the observable $A_0$ and the modified CHSH expression $S_{\text{CHSH}}$ in Eq.~\eqref{eq:s_chsh}. In general, this monogamy relation can be written as:
    \begin{equation}
        \begin{split}
        I_{\text{MONO}}& :=c_1\langle A_0\rangle+c_2 \big (\langle A_0 B_{0,0} \rangle + \langle A_0 B_{1,0} \rangle + \langle A_1 B_{0,1} \rangle - \langle A_1 B_{1,1} \rangle\big),\\
            &\|B_{y,x=0}-B_{y,x=1}\| \leq 2\kappa,\; \forall y\in\{0,1\}.
        \end{split}
    \end{equation}
    for some coefficients $c_1, c_2 \in \mathbb{R}$.
It has been recently shown in~\cite{PAC+25} that in the $(2,2;m,2)$ Bell scenario for any $m \geq 2$, the maximal quantum violation can always be achieved using a two-qubit system, and moreover, the optimal state can be taken to be pure. Without loss of generality, we assume the state that achieves the maximal quantum value of the monogamy relation is $|\psi\>=\cos\theta |00\>+\sin \theta |11\>$ for some $\theta\in(0,\frac{\pi}{4}]$, the corresponding density matrix is then $\rho:=|\psi\>\<\psi|$: 
    \begin{equation}\label{eq:rho}
    \rho=\frac{1}{4}\left(\mathds{1} \otimes \mathds{1}+\cos2\theta\sigma_z\otimes \mathds{1}+\cos2\theta \mathds{1} \otimes \sigma_z +\sum_{i,j} T_{ij} \sigma_{i} \otimes \sigma_{j}\right)
    \end{equation}
    where the $3\times3$ correlation matrix $T$ satisfies $T_{xx} = \sin2\theta$, $T_{yy} = -\sin2\theta$, $T_{zz} = 1$, and all other entries are zero.
Without loss of generality, the observables $A_x$ and $B_{y,x}$ can be parameterized by unit Bloch vectors as:
    \begin{equation}
        A_x=\vec{a}_x \cdot \vec{\sigma},\; B_{y,x}= \vec{b}_{y,x} \cdot \vec{\sigma},\;\text{ with } \|\vec{a}_x\|=1,\; \|\vec{b}_{y,x}\|=1.
    \end{equation}
    where $\vec{\sigma}:=[\sigma_x,\sigma_y,\sigma_z]$. For this state, the expectation value of the observable $A_0$ satisfies the upper bound:
    \begin{equation}\label{eq_a0}
        \<A_0\>\leq \cos2\theta,
    \end{equation}
    and the maximal value of the modified CHSH expression $S_{\text{CHSH}}$ is given by:
    \begin{equation}\label{eq:s_up1}
    S_{\text{CHSH}}=\max_{\substack{\vec{a}_0,\vec{a}_1,\vec{b}_{0,0},\vec{b}_{0,1},\vec{b}_{1,0},\vec{b}_{1,1}, \\\|\vec{a}_0\|=1,\; \|\vec{a}_1\|=1, \\\|\vec{b}_{0,0}-\vec{b}_{0,1}\|\leq 2\kappa, \\ \|\vec{b}_{1,0}-\vec{b}_{1,1}\|\leq 2\kappa.}}\;
    \vec{a}_0\cdot T(\vec{b}_{0,0}+\vec{b}_{1,0})+\vec{a}_1\cdot T(\vec{b}_{0,1}-\vec{b}_{1,1}).
    \end{equation}
   Now, for $y \in \{0,1\}$, define the vectors:
    \begin{equation}
        \vec{c}_y=\frac{1}{2}(\vec{b}_{y,0}+\vec{b}_{y,1}),\;  \vec{d}_y=\frac{1}{2}(\vec{b}_{y,0}-\vec{b}_{y,1}).
    \end{equation}
    Then these vectors satisfy:
    \begin{equation}\label{eq:con_c_d}
        \vec{c}_y\cdot\vec{d}_y=0,\; \|\vec{c}_y\|^2+\|\vec{d}_y\|^2=1,\; \text{ and } \|\vec{d}_y\|\leq \kappa,\forall y\in\{0,1\}.
    \end{equation}
    Using this, the optimization problem in Eq.~\eqref{eq:s_up1} becomes:
    \begin{equation}\label{eq:s_up2}
    \begin{split}    S=&\max_{\substack{\vec{a}_0,\vec{a}_1,\vec{c}_{0},\vec{c}_{1},\vec{d}_{0},\vec{d}_{1}, \\ \|\vec{a}_0\|=1,\; \|\vec{a}_1\|=1, \\ \vec{c}_0\cdot\vec{d}_0=0,\; \|\vec{c}_0\|^2+\|\vec{d}_0\|^2=1, \\ \vec{c}_1\cdot\vec{d}_1=0,\; \|\vec{c}_1\|^2+\|\vec{d}_1\|^2=1,\\
    \|\vec{d}_0\|\leq \kappa,\; \|\vec{d}_1\|\leq \kappa.}}\;
    \vec{a}_0\cdot T(\vec{c}_{0}+\vec{c}_{1}+\vec{d}_{0}+\vec{d}_{1})+\vec{a}_1\cdot T(\vec{c}_{0}-\vec{c}_{1}-\vec{d}_{0}+\vec{d}_{1})\\
     \leq &\max_{\substack{\vec{a}_0,\vec{a}_1,\vec{c}_{0},\vec{c}_{1},\\\|\vec{a}_0\|=1,\; \|\vec{a}_1\|=1,\\ \sqrt{1-\kappa^2}\leq \|\vec{c}_0\|\leq 1,\\ \sqrt{1-\kappa^2}\leq \|\vec{c}_1\|\leq 1.}}\;  \vec{a}_0\cdot T(\vec{c}_{0}+\vec{c}_{1})+\vec{a}_1\cdot T(\vec{c}_{0}-\vec{c}_{1}) + \max_{\substack{\vec{a}_0,\vec{a}_1,\vec{d}_{0},\vec{d}_{1},\\
    \\ \|\vec{a}_0\|=1,\; \|\vec{a}_1\|=1,\\\|\vec{d}_0\|\leq \kappa,\; \|\vec{d}_1\|\leq \kappa.
    }}\; \quad \;\;\; \vec{a}_0\cdot T(\vec{d}_{0}+\vec{d}_{1})+\vec{a}_1\cdot T(\vec{d}_{1}-\vec{d}_{0})\\
    \leq & \max_{\substack{\vec{c}_{0},\vec{c}_{1},\\ \sqrt{1-\kappa^2}\leq \|\vec{c}_0\|\leq 1,\\ \sqrt{1-\kappa^2}\leq \|\vec{c}_1\|\leq 1.}}\;  \|T(\vec{c}_{0}+\vec{c}_{1})\|+\|T(\vec{c}_{0}-\vec{c}_{1})\| + \max_{\substack{\vec{d}_{0},\vec{d}_{1},\\
    \\\|\vec{d}_0\|\leq \kappa,\; \|\vec{d}_1\|\leq \kappa.
    }}\; \quad \;\;\; \| T(\vec{d}_{0}+\vec{d}_{1})\|+\|T(\vec{d}_{1}-\vec{d}_{0})\|.\\
    \end{split}
    \end{equation}
    In the first inequality of Eq.~\eqref{eq:s_up2}, the optimization is split into two separate problems. This inequality can be saturated when the constraints in Eq.~\eqref{eq:con_c_d} hold, and when the optimal directions $\vec{a}_x^{(1)}$ and $\vec{a}_x^{(2)}$ from both parts of the optimization coincide, i.e., $\vec{a}_x^{(1)} = \vec{a}_x^{(2)}$ for all $x \in \{0,1\}$.
In the second inequality of Eq.~\eqref{eq:s_up2}, the optimal directions are chosen as:
\begin{equation}
\begin{split}
    \vec{a}_0^{*(1)}=\frac{T(\vec{c}_{0}+\vec{c}_{1})}{\|T(\vec{c}_{0}+\vec{c}_{1})\|},\;  \vec{a}_1^{*(1)}=\frac{T(\vec{c}_{0}-\vec{c}_{1})}{\|T(\vec{c}_{0}-\vec{c}_{1})\|};\\
    \vec{a}_0^{*(2)}=\frac{T(\vec{d}_{0}+\vec{d}_{1})}{\|T(\vec{d}_{0}+\vec{d}_{1})\|},\;  \vec{a}_1^{*(2)}=\frac{T(\vec{d}_{1}-\vec{d}_{0})}{\|T(\vec{d}_{1}-\vec{d}_{0})\|}.
\end{split}
\end{equation}
Let $\vec{t}_1$ and $\vec{t}_2$ denote the eigenvectors of $T$ corresponding to its largest and second-largest eigenvalues, respectively. Then, saturation of Eq.~\eqref{eq:s_up2} requires:
\begin{equation}\label{eq:con_paral}
\begin{split}
    (\vec{c}_{0}+\vec{c}_{1}) / \! / (\vec{d}_{0}+\vec{d}_{1}) / \! / \vec{t}_1;\\
    (\vec{c}_{0}-\vec{c}_{1}) / \! / (\vec{d}_{1}-\vec{d}_{0}) / \! / \vec{t}_2.
\end{split}
\end{equation}
For the state in Eq.~\eqref{eq:rho}, $\vec{t}_1$ is aligned with the $z$-axis and $\vec{t}_2$ lies in the $x$-$y$ plane. Without loss of generality, assume $\vec{t}_2$ is aligned along the $x$-axis. Therefore, parameterizing $\vec{b}_{y,x}$ in the $x$-$z$ plane as $\vec{b}_{y,x}=[\cos\gamma_{y,x},0,\sin\gamma_{y,x}],\forall x,y\in\{0,1\}$, the conditions in Eq.~\eqref{eq:con_paral} imply:
\begin{equation}
\begin{split}
    \sin\gamma_{0,0}=\sin\gamma_{1,0}=:\sin\alpha,\; \cos\gamma_{0,0}=-\cos\gamma_{1,0}=:-\cos\alpha;\\
    \sin\gamma_{0,1}=\sin\gamma_{1,1}=:\sin\beta,\; \cos\gamma_{0,0}=-\cos\gamma_{1,0}=:-\cos\beta.
\end{split}
\end{equation}
In other words, we have:
\begin{equation}
    \begin{split}
        &\vec{b}_{0,0}=[-\cos\alpha,0,\sin\alpha],\;\vec{b}_{0,1}=[\cos\beta,0,\sin\beta],\;\vec{b}_{1,0}=[\cos\alpha,0,\sin\alpha],\;\vec{b}_{1,1}=[-\cos\beta,0,\sin\beta];\;\\
    \end{split}
\end{equation}
and
\begin{equation}
\begin{split}
            &\vec{c}_0+\vec{c}_1=[0,0,\sin\alpha+\sin\beta],\;\vec{c}_0-\vec{c}_1=[\cos\beta-\cos\alpha,0,0],\;\\
            &\vec{d}_0+\vec{d}_1=[0,0,\sin\alpha-\sin\beta],\; \vec{d}_1-\vec{d}_0=[\cos\alpha+\cos\beta,0,0].
\end{split}
\end{equation}
In addition, the constraints $\|\vec{b}_{0,0}-\vec{b}_{0,1}\|\leq 2\kappa$, and $\|\vec{b}_{1,0}-\vec{b}_{1,1}\|\leq 2\kappa$ yield:
\begin{equation}
    \cos(\alpha+\beta)\leq 2\kappa^2-1.
\end{equation}
Thus, the right-hand side of Eq.~\eqref{eq:s_up2} simplifies to:
\begin{equation}\label{eq:s_up3}
\begin{split}
    S & \leq  \max_{\substack{\alpha,\beta,\\
    \\ \cos(\alpha+\beta)\leq 2\kappa^2-1.
    }}\;2\big(\sin\alpha + \sin2\theta \cos\beta\big)\\
    & \leq \max_{\substack{\gamma,\\\cos{\gamma}\leq 2\kappa^2-1}}\;\max_{\alpha} \; 2\big(\sin\alpha + \sin2\theta \cos(\gamma-\alpha)\big)\\
    &\leq \max_{\substack{\gamma,\\\cos{\gamma}\leq 2\kappa^2-1}}\;2\sqrt{1+\sin^2 2\theta + 2\sin 2\theta\sin\gamma}\\
    & =\begin{cases}
        2\sqrt{1+\sin^2 2\theta + 4\kappa\sqrt{1-\kappa^2}\sin2\theta},\; & \text{ when } \kappa\in[0,\frac{1}{\sqrt{2}}];\\
        2(1+\sin2\theta), & \text{ when } \kappa\in[\frac{1}{\sqrt{2}},1).
    \end{cases}
\end{split}
\end{equation}
Several conclusions can be drawn from this derivation. First, the quantum value of the inequality $S_{\text{CHSH}}$ is obtained by optimizing the parameter $\theta$ for the state $|\psi\rangle$ in Eq.~\eqref{eq:s_up3}. That is, the quantum value of the modified CHSH test in Eq.~\eqref{eq:s_chsh} is given by:
\begin{equation}\label{eq:sq}
    S_{q,\kappa}=\begin{cases}
        2\sqrt{2}(\kappa+\sqrt{1-\kappa^2}),\; & \text{ when } \kappa\in[0,\frac{1}{\sqrt{2}}];\\
        4, & \text{ when } \kappa\in[\frac{1}{\sqrt{2}},1).
    \end{cases}
\end{equation}
This maximum is achieved by the state $|\psi\rangle = \frac{1}{\sqrt{2}}(|00\rangle + |11\rangle)$ and the observables:
\begin{align}
    &\text{ when } \kappa\in[0,\frac{1}{\sqrt{2}}],\quad 
    \begin{cases}
        &A_0=\sigma_z,\; A_1=\sigma_x;\\
        &B_{0,0}=\frac{-\kappa+\sqrt{1-\kappa^2}}{\sqrt{2}}\sigma_x+\frac{\kappa+\sqrt{1-\kappa^2}}{\sqrt{2}}\sigma_z,\\
        &B_{0,1}=\frac{\kappa+\sqrt{1-\kappa^2}}{\sqrt{2}}\sigma_x+\frac{-\kappa+\sqrt{1-\kappa^2}}{\sqrt{2}}\sigma_z,\\
        &B_{1,0}=\frac{\kappa-\sqrt{1-\kappa^2}}{\sqrt{2}}\sigma_x+\frac{\kappa+\sqrt{1-\kappa^2}}{\sqrt{2}}\sigma_z,\\
        &B_{1,1}=\frac{-\kappa-\sqrt{1-\kappa^2}}{\sqrt{2}}\sigma_x+\frac{-\kappa+\sqrt{1-\kappa^2}}{\sqrt{2}}\sigma_z.
    \end{cases}\\
    & \text{ when } \kappa\in[\frac{1}{\sqrt{2}},1),\quad 
    \begin{cases}
        &A_0=\sigma_z,\; A_1=\sigma_x;\\
        &B_{0,0}=\sigma_z,\;B_{0,1}=\sigma_x,\\
        &B_{1,0}=\sigma_z,\;B_{1,1}=-\sigma_x.
    \end{cases}
\end{align}
On the other hand, if the underlying state is a partially entangled state, the observed quantum value of the modified CHSH test in Eq.~\eqref{eq:s_chsh} is given by Eq.~\eqref{eq:s_up3}, while the expectation value $\langle A_0 \rangle$ saturates the bound in Eq.~\eqref{eq_a0}. Combining Eqs.~\eqref{eq_a0} and \eqref{eq:s_up3}, the relation between $\langle A_0 \rangle$ and $S$ is:
\begin{equation}\label{eq:A0_bound}
    \<A_0\>= \begin{cases}
        \frac{1}{2}\sqrt{4-\big( \sqrt{S^2-(2-4\kappa^2)^2} -4\kappa\sqrt{1-\kappa^2}\big)^2}, & \text{ when } \kappa\in[0,\frac{1}{\sqrt{2}}];\\
        \frac{1}{2}\sqrt{S(4-S)}, & \text{ when } \kappa\in[\frac{1}{\sqrt{2}},1).\\
    \end{cases}
\end{equation}
Note that applying the local transformation $A_x \leftrightarrow -A_x, \; B_{y,x} \leftrightarrow -B_{y,x}$ for all $x, y \in \{0,1\}$ leaves the CHSH test in Eq.~\eqref{eq:s_chsh} invariant. Hence, the bound in Eq.~\eqref{eq:A0_bound} also applies to $-\langle A_0 \rangle$:
\begin{equation}\label{eq:negA0_bound}
    -\<A_0\>= \begin{cases}
        \frac{1}{2}\sqrt{4-\big( \sqrt{S^2-(2-4\kappa^2)^2} -4\kappa\sqrt{1-\kappa^2}\big)^2}, & \text{ when } \kappa\in[0,\frac{1}{\sqrt{2}}];\\
        \frac{1}{2}\sqrt{S(4-S)}, & \text{ when } \kappa\in[\frac{1}{\sqrt{2}},1).\\
    \end{cases}
\end{equation}
Therefore, combining the bounds from Eqs.~\eqref{eq:A0_bound} and~\eqref{eq:negA0_bound}, for any $\kappa \in [0,1)$ and any $S \in [S_{c,\kappa}, S_{q,\kappa}]$, where $S_{c,\kappa} = 2 + 2\kappa$ and $S_{q,\kappa}$ is given in Eq.~\eqref{eq:sq}, the guessing probability ${P}_{g}^{(q)}(S,\kappa,x^*=0)$ is bounded by:
\begin{equation}\label{eq:pg_boud1}
\begin{split}
    & {P}_{g}^{(q)}(S,\kappa,x^*=0)= \max_{\substack{\{\bm{p}_{AB|XY}^a\},\; \{q_a\},\\ q_a=0,\forall a,\; \sum_a q_a=1,\\ \{\bm{p}_{AB|XY}^a\}\in\mathcal{Q},\forall a,\\
    \{\sum_a q_a \bm{p}_{AB|XY}^a\} \text{ achieving $S$ for Eq.~\eqref{eq:s_chsh}}.}}\sum_{a} q_a \sum_{b}{p}^a_{AB|XY}(a,b|x^*=0,y) \\
    \geq & \max_{\substack{ \{q_a\},\\ q_a=0,\forall a,\; \sum_a q_a=1.}} \; \bigg( \max_{\substack{ \{\bm{p}_{AB|XY}^{a=0}\}\in\mathcal{Q}\\
    \{\bm{p}_{AB|XY}^{a=0}\} \text{ achieving $S$ for Eq.~\eqref{eq:s_chsh}}.}} \;q_{a=0} \frac{1+\<A_0\>}{2} + \max_{\substack{ \{\bm{p}_{AB|XY}^{a=1}\}\in\mathcal{Q}\\
    \{\bm{p}_{AB|XY}^{a=1}\} \text{ achieving $S$ for Eq.~\eqref{eq:s_chsh}}.}} \;q_{a=1} \frac{1-\<A_0\>}{2}\bigg)\\
    = &  \begin{cases}
        \frac{1}{2}+\frac{1}{4}\sqrt{4-\big( \sqrt{S^2-(2-4\kappa^2)^2} -4\kappa\sqrt{1-\kappa^2}\big)^2}, & \text{ when } \kappa\in[0,\frac{1}{\sqrt{2}}];\\
        \frac{1}{2}+\frac{1}{4}\sqrt{S(4-S)}, & \text{ when } \kappa\in[\frac{1}{\sqrt{2}},1).
        \end{cases}
\end{split}
\end{equation}
Let us denote the expression in the final line of Eq.~\eqref{eq:pg_boud1} as $\overline{P}_{g}^{(q)}(S,\kappa,x^*=0)$. When $S = S_{c,\kappa}$, the system behaves classically, and Eve can guess the outcome of $A_0$ with certainty, i.e., the guessing probability equals $1$. In contrast, the bound derived in Eq.~\eqref{eq:pg_boud1} yields: 
\begin{equation}
\overline{P}_{g}^{(q)}(S=S_{c,\kappa},\kappa,x^*=0) = \begin{cases}
    \frac{1}{2}+\frac{1}{2} \sqrt{1-\left(\sqrt{2 \kappa+5 \kappa^2-4 \kappa^4}-2 \kappa \sqrt{1-\kappa^2}\right)^2}, & \text{ when } \kappa\in[0,\frac{1}{\sqrt{2}}];\\
    \frac{1}{2}+\frac{1}{2} \sqrt{1-\kappa^2}, & \text{ when } \kappa\in[\frac{1}{\sqrt{2}},1).\\
\end{cases}
\end{equation} 
which is strictly less than 1 except when $\kappa=0$. This implies that $\overline{P}_{g}^{(q)}(S,\kappa,x^*=0)$ is not a tight bound on ${P}_{g}^{(q)}(S,\kappa,x^*=0)$; Eve may implement more effective strategies to improve her guessing probability.

One such strategy involves preparing distributions $\{\bm{p}_{AB|XY}^a\}$ and weights $\{q_a\}$ such that some behaviors achieve a Bell value lower than $S$ (e.g., $S_{c,\kappa}$, where Eve can guess with certainty), while others achieve a higher value (e.g., $S_{\kappa}^* > S$). A convex mixture of these strategies can still reproduce the observed value $S$. In this case, the overall guessing probability becomes the corresponding convex combination of 1 and $\overline{P}{g}^{(q)}(S=S_{\kappa}^*,\kappa,x^*=0)$, weighted by $\{q_a\}$.
To determine the best possible guessing probability from this convex-mixture strategy, the optimal value $S_{\kappa}^*$ is found as the maximizer of the following optimization:
\begin{equation}
	 \max_{S\in[S_{c,\kappa},S_{q,\kappa}]} \quad \frac{\overline{P}_{g}^{(q)}(S,\kappa,x^*=0) - 1}{S- S_{c,\kappa}} .
\end{equation}
In other words, $S_{\kappa}^*$ is the unique real solution within $[S_{c,\kappa},S_{q,\kappa}]$ of the equation:
\begin{equation}
    \frac{\partial \overline{P}_{g}^{(q)}(S,\kappa,x^*=0)}{\partial S}\left(S-S_{c, \kappa}\right)=\overline{P}_{g}^{(q)}(S,\kappa,x^*=0)-1.
\end{equation}
Thus, under this convex-mixture strategy, a tighter lower bound on the guessing probability is given by the piecewise function:
\begin{equation}\label{eq:pg_boud2}
    {P}_{g}^{(q)}(S,\kappa,x^*=0)\geq \begin{cases}
    \overline{P}_{g}^{(q)}(S,\kappa,x^*=0), & \text{ when } S\in[S_{\kappa}^*,S_{q,\kappa}];\\
    \frac{\overline{P}_{g}^{(q)}(S=S_{\kappa}^*,\kappa,x^*=0)-1}{S_{\kappa}^*-S_{c,\kappa}} S+\frac{S_{\kappa}^*-S_{c,\kappa}\overline{P}_{g}^{(q)}(S=S_{\kappa}^*,\kappa,x^*=0)}{S_{\kappa}^*-S_{c,\kappa}}, & \text{ when } S\in[S_{c,\kappa},S_{\kappa}^*].\\
\end{cases}
\end{equation}
We remark that when $\kappa = 0$, we recover the well-known result for the guessing probability in the standard CHSH Bell test without any information leakage~\cite{PAM+10}.

A numerical upper bound on the guessing probability can be obtained using the Navascu\'{e}s-Pironio-Ac\'{i}n (NPA) hierarchy of semi-definite programs~\cite{NPA1, NPA2}. Figure~\ref{fig_pg} plots the analytical curve derived in Eq.~\eqref{eq:pg_boud2}, based on our candidate strategy, alongside the upper bound obtained from level-$3$ of the NPA hierarchy. As shown in the figure, the analytical bound closely matches the numerical upper bound, confirming the tightness and optimality of our strategy. The corresponding min-entropy, defined as $H_{\min}(S,\kappa,x^*=0) = -\log_2 \left[{P}_{g}^{(q)}(S,\kappa,x^*=0)\right]$, is shown in Figure~\ref{fig_hmin}.

\end{proof}

\begin{figure}[H]
	\centering
	\subfigure[]{
		\begin{minipage}[t]{0.48\linewidth}
			\centering
			\includegraphics[width=\textwidth]{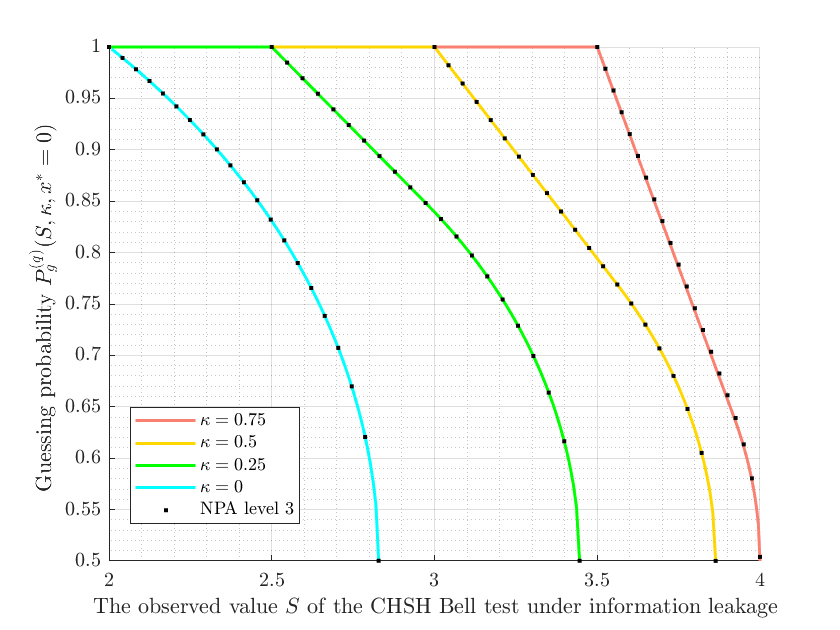}
			\label{fig_pg}
		\end{minipage}%
	}%
	\subfigure[]{
		\begin{minipage}[t]{0.48\linewidth}
			\centering
			\includegraphics[width=\textwidth]{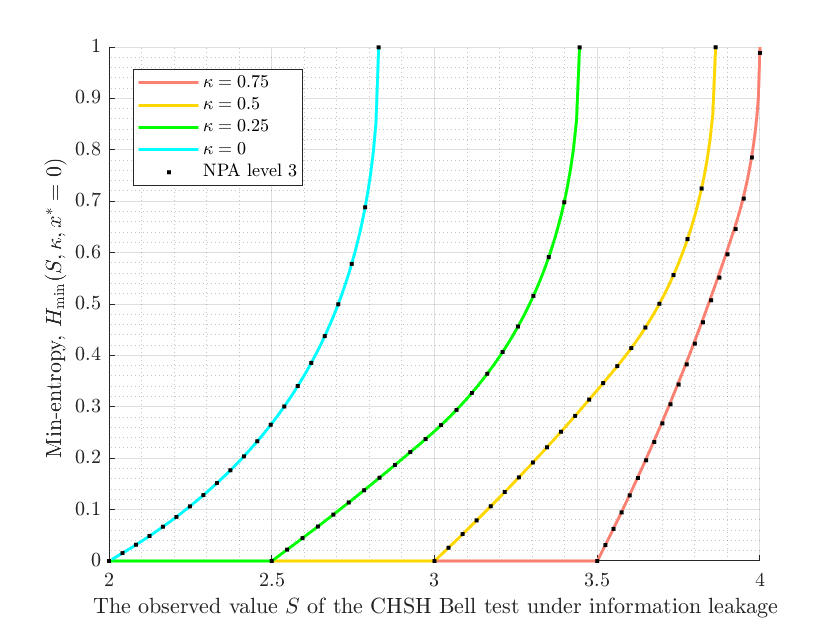}
			\label{fig_hmin}
		\end{minipage}%
	}
	\centering
    \caption{(a) Solid lines: Lower bound of the guessing probability ${P}_{g}^{(q)}(S,\kappa,x^*=0)$ as a function of the observed CHSH value $S$ under $\kappa$ bits of input information leakage, as derived in Eq.~\eqref{eq:pg_boud2}, for $\kappa = 0, 0.25, 0.5, 0.75$.  Black dots: Upper bounds on the guessing probability obtained using the level-3 NPA hierarchy~\cite{NPA1,NPA2}.  (b) Corresponding min-entropy $H_{\min}(S,\kappa,x^*=0) = -\log_2({P}_{g}^{(q)}(S,\kappa,x^*=0))$ as a function of $S$.}
\end{figure}

We also numerically study the scenario in which the input choices of both Alice and Bob are subject to information leakage, quantified by the parameters $\kappa_A$ and $\kappa_B$, respectively. In the figure below Fig.~\ref{fig_pg_k}, we compare the resulting guessing probability in this two-way information-leakage scenario with the one-way information-leakage case analyzed above. For the two-way leakage case, we set $\kappa_A=\kappa_B=\kappa$ and plot the guessing probability as a function of $\kappa$, for a fixed observed CHSH value $S = 2\sqrt{2}$. The corresponding curve is obtained numerically using the NPA hierarchy~\cite{NPA1,NPA2} (codes can be found in~\cite{github}). For the one-way information-leakage case, where only Alice's input information is leaked, i.e., $\kappa_A=\kappa$, $\kappa_B=0$, the guessing probability is given by ${P}_{g}^{(q)}(S=2\sqrt{2},\kappa,x^*=0)$, as derived above.

\begin{figure}[H]
	\centering
		\begin{minipage}[t]{0.6\linewidth}
			\centering
			\includegraphics[width=\textwidth]{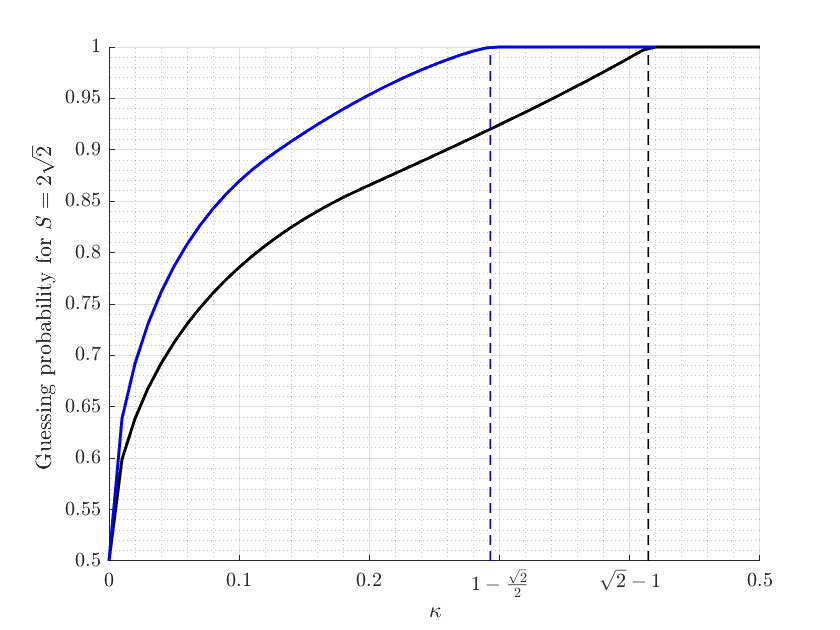}
		\end{minipage}
    \caption{The guessing probability as a function of $\kappa$, which quantifies the amount of input information leakage, for a fixed observed CHSH value $S = 2\sqrt{2}$. The black line corresponds to the one-way information-leakage scenario, where only Alice's input information is leaked, i.e. $\kappa_A=\kappa$, $\kappa_B=0$, and the guessing probability is ${P}_{g}^{(q)}(S=2\sqrt{2},\kappa,x^*=0)$ derived above. The blue line corresponds to the two-way information-leakage scenario, where both Alice's and Bob's input information are leaked, with $\kappa_A=\kappa_B=\kappa$; the associated guessing probability curve is obtained numerically using the NPA hierarchy.}
    \label{fig_pg_k}
\end{figure}

The toolkit Moment~\cite{GA24}, the modeler CVX~\cite{CVX}, and the solver MOSEK~\cite{Mosek} are used for numerical calculation.

\end{document}